\documentclass[a4paper,11pt,onecolumn]{IEEEtran} 

\usepackage{amssymb,amsmath,amsfonts,color,mathrsfs,amsthm,verbatim}
\usepackage{algorithm}
\usepackage{enumerate}
\usepackage{graphicx}
\usepackage[noend]{algpseudocode}

\makeatletter
 
\newcommand{\dm}{d_{\rm{max}}}
\newcommand{\kr}{\left \lceil \frac{k}{r} \right \rceil}
\newcommand{\lkr}{\left \lceil \frac{k}{r} \right \rceil}

\newcommand{\com}{\hbox{, }}
\newcommand{\hand}{\hbox{ and }}

\newcommand{\apart}{(n,k,d,r,\delta)}

\newcommand{\atms}{\apart \hbox{-matroids}}

\newcommand{\cflat}{\mathcal{Z}}

\newcommand{\cl}{\mathrm{cl}}
\newcommand{\coatom}{\mathrm{coA}_\cflat}

\newtheorem {corollary}{Corollary}[section]
\newtheorem {lemma}{Lemma}[section]
\newtheorem {thm}{Theorem}[section]
\newtheorem {proposition}{Proposition}[section]
\newtheorem {exam}{Example}[section]
\newtheorem {definition}{Definition}[section]
\begin{document}
\pagestyle{plain}

\title{On the Combinatorics of Locally Repairable Codes via Matroid Theory}
\author{Thomas Westerb\"ack\thanks{T. Westerb\"ack, T.~Ernvall, and C. Hollanti  are with the Department of Mathematics and Systems Analysis, 
%P.O. Box 11100, FI-00076 
Aalto University, Finland.}, Ragnar Freij-Hollanti\thanks{R. Freij-Hollanti is with the Department of Communications and Networking, Aalto University, Finland.}, Toni Ernvall, and  Camilla Hollanti\thanks{E-mails: \{thomas.westerback, ragnar.freij, toni.ernvall, camilla.hollanti\}@aalto.fi.}
\thanks{The research of R. Freij-Hollanti is partially supported by the Finnish Academy of Science and Letters. The research of  C. Hollanti is supported by the Academy of Finland grants \#276031, \#282938, and \#283262, and by Magnus Ehrnrooth Foundation, Finland. The support from the European Science Foundation under the ESF COST Action IC1104 is
also gratefully acknowledged.}\thanks{Preliminary and partial results of this paper were presented at the 2014 IEEE Information Theory Workshop (ITW) in Hobart, Tasmania \cite{tassie}.}
}
\maketitle

%\thanks{The research of T. Ernvall is supported in part  by the Academy of Finland grant \#131745.}
%\thanks{The research of T. Westerb{\"a}ck is funded by ...}
%\thanks{The research of C. Hollanti is funded by...}
%\thanks{Part of this work appeared at ...}

%\author{
%\authorblockN{Thomas Westerb{\"a}ck$^*$}
%\authorblockA{$^*$Department of Mathematics and Systems Analysis\\ Aalto University, P.O. Box 11100\\ FI-00076 Aalto, Finland \\ (e-mails: \{firstname.lastname@aalto.fi\})}
%\and
%\authorblockN{Toni Ernvall$^\dag$ and Camilla Hollanti$^*$}
%\authorblockA{$^\dag$Turku Centre for Computer Science, Turku, Finland\\ \&  Department of Mathematics and Statistics\\ FI-20014 University of Turku, Finland\\ (e-mail: tmernv@utu.fi)}
%}

\maketitle
\begin{abstract}
This paper provides a link between matroid theory and locally repairable codes (LRCs) that are either linear or more generally almost affine. Using this link, new results on both LRCs and matroid theory are derived. The parameters $(n,k,d,r,\delta)$ of LRCs are generalized to matroids, and the matroid analogue of the generalized Singleton bound in [P. Gopalan \emph{et al.}, ``On the locality of codeword symbols,'' \emph{IEEE Trans. Inf. Theory}] for linear LRCs is given for matroids. It is shown that the given bound is not tight for certain classes of parameters,  implying a nonexistence result for the corresponding locally repairable almost affine codes, that are coined \emph{perfect} in this paper. 

Constructions of classes of matroids with a large span of  the parameters $(n,k,d,r,\delta)$ and the corresponding local repair sets are given. Using these matroid constructions, new LRCs are constructed with prescribed parameters. The existence results on linear LRCs and the nonexistence results on almost affine LRCs given in this paper strengthen the nonexistence and existence results on perfect linear LRCs given in [W. Song \emph{et al.}, ``Optimal locally repairable codes,'' \emph{IEEE J. Sel. Areas Comm.}].
\end{abstract}

\section{Introduction}
Due to the ever-growing need for more efficient and scalable  systems for cloud storage and data storage in general, distributed storage has become an increasingly important ingredient in many data systems. In their seminal paper \cite{dimakis}, Dimakis \emph{et al.} introduced network coding techniques for  large-scale  distributed storage systems such as data centers, cloud storage, peer-to-peer storage systems and storage in wireless networks. These techniques can, for example, considerably improve the storage efficiency compared to traditional storage techniques such as replication and erasure coding.  

Failing devices are not uncommon in large-scale distributed storage systems \cite{ghemawat03}. A central problem for this type of storage is therefore to design codes that have good distributed repair properties. Several cost metrics and related tradeoffs \cite{dimakis,me,sasidharan,tian,toni,goparaju} are studied in the literature, for example \emph{repair bandwidth} \cite{dimakis,me}, \emph{disk-I/O} \cite{diskIO}, and \emph{repair locality} \cite{Gopalan,Oggier,Simple}. In this paper repair locality is the subject of interest. 

The notion of a \emph{locally repairable code} (LRC) was introduced in \cite{LRCpapailiopoulos}, and such repair-efficient codes are already used in existing distributed storage systems, \emph{e.g.}, in the Hadoop Distributed File System \emph{RAID} used by Facebook and Windows Azure Storage \cite{LRCmatroid}. There are two notions of \emph{symbol locality} considered in the literature: information locality only requires information symbols to be locally repairable, while all-symbol locality requires this to be true for all code symbols. The subject of interest in this paper is the all-symbol locality.

It is well-known that nonlinear codes often achieve better performance than linear ones, \emph{e.g.}, in the context of coding rates for error-correcting codes and maximal throughput for network codes. Almost affine codes were introduced in \cite{simonis98} as a generalization of linear codes. This class of codes contains codes over arbitrary alphabet size, not necessarily prime power. In this paper, we are studying LRCs in the generality of almost affine codes.

We will consider five key invariants $(n,k,d,r,\delta)$ of locally repairable codes. The technical definitions are given in Section~\ref{sec:LRC}, but in short, a good code should have large rate $k/n$ as well as high global and local failure tolerance $d$ and $\delta$, respectively. In addition, it is desirable to have small $r$, which will determine the maximum number of nodes that have to be contacted for repair within a ``local'' repair set.

In this paper, our main tools for analyzing LRCs come from matroid theory. This is a branch of algebraic combinatorics with natural links to a great number of different topics, \emph{e.g.}, to coding theory, graph theory, matching theory and combinatorial optimization. Matroids were introduced in \cite{whitney35} in order to abstractly capture properties analogous to linear independence in vector spaces and independence in graphs. Since its introduction, matroid theory has been successfully used to solve problems in many areas of mathematics and computer science.  Matroid theory and the theory of linear codes are closely related since every matrix over a field defines a matroid. Despite this fact, until rather recently matroid theory has only played a minor part in the development of coding theory. One pioneering work in this area is the paper by Greene from 1976 \cite{greene76}. In this paper he describes how the weight enumerator of a linear code $C$ is determined by the Tutte polynomial of the associated matroid of $C$. Using this result, Greene gives an elegant proof of the MacWilliams identity \cite{macwilliams63}. Generalizations of these results have then been presented in several papers, for example in \cite{barg97, britz10}. Another important instance of matroidal methods in coding theory is the development of a decomposition theory of binary linear codes \cite{kashyap08}.   Today, matroid theory also plays an important role in information theory and coding theory, for example in the areas of network coding, secret sharing, index coding, and information inequalities \cite{dougherty07, marti07, rouayheb10}. In this paper, while our main goal is investigating almost affine LRCs with the aid of matroid theory, ideas from the theory of LRCs will also be utilized to acquire new results in matroid theory.

\subsection{Related work}

One of the most classical theorems in coding theory is the Singleton bound, discussed in Section~\ref{Sec:Singleton}~\cite{Singleton}. Its classical version bounds the minimum distance $d$  of a code from above in terms of the length $n$ and dimension $k$. Recent work  sharpens the bound in terms of the local parameters $(r,\delta)$~\cite{Gopalan, prakash12, prakash14, wang15}, as well as in terms of other parameters~\cite{ LRCpapailiopoulos, cadambe13, rawat14}.

There are different constructions of LRCs that are optimal in the sense that they achieve a generalized Singleton bound, \emph{e.g.} \cite{LRCmatroid, prakash12, silberstein13, song14, tamo14}. Song \emph{et al.} \cite{song14}  investigate for which  parameters $(n,k,r,\delta)$ there exists a linear  LRC with all-symbol locality and minimum distance $d$ achieving the generalized Singleton bound from~\cite{prakash12}. The parameter set $(n,k,r,\delta)$ is divided into eight different classes. In four of these classes it is proven that there are linear LRCs achieving the bound, in two of these classes it is proven that there are no linear LRCs achieving the bound, and the existence of linear LRCs achieving the bound in the remaining two cases is an open question.  Independently to the research in this paper, Wang and Zhang used linear programming approaches to strengthen these results when $\delta=2$~\cite{wang15}.

%A matroid representable by a code can in general be represented by many different codes. Further, there are matroids which can be represented by almost affine codes but not by any linear code \cite{simonis98}. It has been proven for linear codes that many of their properties are \emph{matroid-invariant}, that is, only depend on their associated matroid.  Examples of matroid invariants  for linear codes are the dimension, the length, and the distributions of supports, weights, higher supports and higher weights \cite{barg97, britz10}. Other properties, such as the covering radius, are not matroid-invariant for linear codes \cite{britz05}. 
It was shown in \cite{LRCmatroid}, that the $r$-locality of a linear LRC is a matroid invariant. This was used in \cite{LRCmatroid} to prove that the minimum distance of a class of linear LRCs achieves a generalized Singleton bound. Moreover, there are several instances of results in the theory of linear codes that have been generalized to all matroids. Examples on how these results can be interpreted for other objects that can represent a matroid, such as graphs, transversals and certain designs can be found in~\cite{britz12}.   

Recently, the present authors have studied locally repairable codes with all-symbol locality  \cite{eka-lrc}. Methods to modify already existing
codes were presented and it was shown that with high probability, a certain random matrix will be a generator matrix for a locally repairable code with a good minimum distance. Constructions were given for three
infinite classes of optimal vector-linear locally repairable codes over an alphabet of small size. The present paper extends and deviates from   this work by studying the combinatorics of LRCs in general and relating LRCs to matroid theory. This allows for the derivation of fundamental bounds for matroids and linear and almost affine LRCs, as well as for the characterization of the matroids achieving this bound.

In this paper, we have chosen to call the codes and matroids achieving the generalized Singleton bound \emph{perfect} instead of optimal, reserving the term optimal for the best existing solution, \emph{i.e.}, for codes achieving a tight bound instead of the (in some cases loose) Singleton bound. See Definition \ref{Def:perfect} and the follow-up footnote for more details. 

\subsection{Contributions and organization}

%In this paper we investigate the combinatorics of LRCs. The main focus is on the connections between matroid theory and LRCs with all-symbol locality $(r,\delta)$. In particular, we are interested in locally repairable codes that are almost affine. This class contains all linear LRCs over finite fields. 

The first contribution of this paper is to extend the definitions of the parameters $(n,k,d,r,\delta)$ in \cite{prakash12} from linear codes to the much larger class of almost affine codes, and to show that these parameters are matroid invariant for all almost affine LRCs. We then proceed to prove the main results of this paper, which can be summarized as follows:
\begin{enumerate}[(i)]
\item A matroid analogue of the generalized Singleton bound in~\cite{prakash12} is given for $(n,k,d,r,\delta)$-matroids, and in particular to all almost affine codes in Theorem \ref{th:bound}. 

\item In Theorem \ref{th:structure-optimal-matroid}, some necessary structural properties are given for an $(n,k,d,r,\delta)$-matroid meeting the generalized Singleton bound.

\item In Theorem~\ref{thm:construction}, a class of matroids is given with different values of the parameters $(n,k,d,r,\delta)$. Simple and explicit constructions of matroids in this class are given in  Theorem \ref{thm:construction} , Theorem \ref{theorem:graph_1}, and Corollary \ref{corollary:graph_2}, and in Examples \ref{exam:graph_construction_1}, \ref{exam:graph_construction_2a} and \ref{exam:graph_construction_2b}. 
%The matroids are constructed via certain ranked subsets of the storage nodes, or via weighted graphs.
% inducing lattices of cyclic flats, using the structural properties mentioned in (ii). 

\item In Section \ref{sec:representability}, we prove that the matroids from Theorem~\ref{thm:construction} are representable over finite fields of large enough size.  Hence we obtain four explicit constructions of linear LRCs with given parameters. The representability is derived by constructing a graph supporting a gammoid isomorphic to the matroid in Theorem~\ref{thm:construction}, and using results on representability of gammoids~\cite{lindstrom73}.
%One benefit with the construction is that the field size $q$ of our construction is comparable with $n$ instead of being a exponential function of $n$ as in the constructions given in \cite{song14}.
%\begin{color}{red}These results are dependent on results on the smallest field size for different MDS-codes, from~\cite{dau14?}

\item Theorem~\ref{thm:max_d} characterizes values of $(n,k,r,\delta)$ for which there exist $(n,k,d,r,\delta)$-matroids meeting the bound (i). In particular, the nonexistence results for linear LRCs in \cite{song14} are extended to the nonexistence of almost affine codes and matroids.
Moreover, in Theorem~\ref{thm:non_existence_LRC} and Theorem~\ref{thm:max_d_LRC}, we settle the existence  in one of the regimes left open in~\cite{song14}, leaving open only a minor subregime of $b>a \geq \lceil \frac{k}{r} \rceil - 1$, where $a=r\lkr -k$ and $b=(r+\delta-1)\lceil\frac{n}{r+\delta-1}\rceil-n$. This complements recent and independent research by Wang and Zhang~\cite{wang15}, where they settle the existence in the subregime $\lceil\frac{n}{r+1}\rceil > b$ and $\delta=2$ using integer programming techniques.

\end{enumerate}

%The method in this paper of using matroid theory, especially the lattice of cyclic flat, in connection with LRCs introduces new tools to use in the research area of linear LRCs and more generally almost affine LRCs. Furthermore, the matroid results in these paper can also give new results on other objects than almost affine code that can be used to represent matroids.  

%In this paper, for certain values of $(d,r)$ we give new explicit constructions of optimal LRCs which are almost affine. These codes have the good property of being over the small field $\mathbb{F}_4$ and not depending on the size of $n$. 

%Section \ref{sec:matroids-codes} gives the basic definitions and properties of $(n,k,d,r,\delta)-$LRCs, almost affine codes and  matroids. In section \ref{sec:(n,k,d,r,delta)-matroid}, we first prove that all the parameters $(n,k,d,r,\delta)$ of an almost affine code are matroid invariants. By using this result, we define $(n,k,d,r,\delta)$-matroids. Moreover, the parameters $(n,k,d,r,\delta)$ of a matroid is connected to its cyclic flats. Using this connection, results on $(n,k,d,r,\delta)$-matroids are given. 
%\begin{color}{red}
%REWRITE\\
%In Section \ref{sec:applications} we use our results on $(n,k,d,r,\delta)$-matroids in order to get results on linear LRCs and more generally on almost affine LRCs. 
%
The proofs of some of the longer theorems and the explicit constructions of matroids with prescribed parameters are given in the Appendix.

\section{Preliminaries} \label{sec:matroids-codes}

\subsection{Parameters $(n,k,d,r,\delta)$ of locally repairable codes} \label{sec:LRC}

In this subsection, we introduce the parameters $(n,k,d,r,\delta)$ defined in \cite{prakash12} for linear locally repairable codes. We extend this definition to the much wider class of almost affine codes, to be introduced in~\ref{Sec:AlmostAffine}. Figure~\ref{fig:cloud}  serves as a visual aid for the technical definitions. The information symbols $(a,b,c,d,e,f)$ are stored on twelve nodes as in the figure. Equivalently, we think of the content of the twelve nodes as a codeword, and of the content of an individual node as a code symbol. Within each of the local clouds (or locality sets), three symbols are enough to determine the other two. Thus, Figure~\ref{fig:cloud} depicts a  $(12,6,3,3,3)$-LRC,  according to the following definitions.

\begin{figure}[htb]
    \centering
    \includegraphics[width=0.6\textwidth]{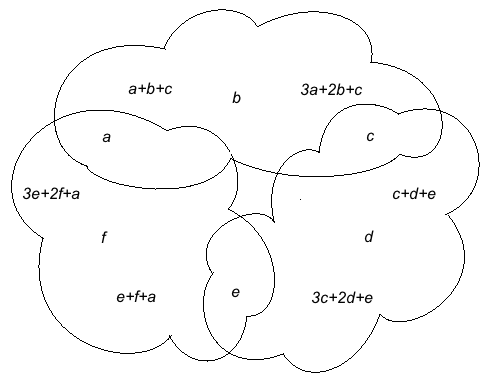}
    \caption{A storage system from a $(12,6,3,3,3)$-LRC.}
    \label{fig:cloud}
\end{figure}

%The $i$th coded symbol $c_i$ in an $\lbrack n, k \rbrack$ linear code is said to have \emph{repair locality} $r$, ($1 \leq r \leq k$), if the code symbol $c_i$ can be recovered by accessing at most $r$  other symbols in the code.  By a \emph{linear} code we mean more precisely a linear code \emph{over a finite field} $\mathbb{F}_q$, with length $n$, dimension $k$, and field size $q$. The concept of $r$-locality for linear codes was generalized to $(r,\delta)$-locality in \cite{prakash12}.

Let  $C \subseteq \mathbb{A}^n$ be a code such that $\lvert C \rvert = \lvert \mathbb{A} \rvert^k$, where $\mathbb{A}$ is a finite set, also referred to as the \emph{alphabet}. For any subset $X = \{i_1,\ldots,i_m\} \subseteq \lbrack n \rbrack = \{1,2,\ldots,n\}$, let $C_X$ denote the \emph{projection} of the code into $\mathbb{A}^{\lvert X \rvert}$, that is
\begin{equation} \label{eq:projection}
C_X = \{(c_{i_1},\ldots,c_{i_m}) : \boldsymbol{c} = (c_1,\ldots,c_n) \in C\}.
\end{equation}
 The code $C_X$ is also called a punctured code in the coding theory literature. The minimum (Hamming) distance $d$ of $C$ can be defined in terms of projections as
\begin{equation} \label{eq:d_via_projection}
d = \min \{|X| : X \subseteq [n] \hbox{ and } |C_{[n]\setminus X}| < |C|  \}. 
\end{equation} 

For $1 \leq r \leq k$ and $\delta \geq 2$, an $(r,\delta)$\emph{-locality set} of $C$ is a subset $S \subseteq \lbrack n \rbrack$ such that
\begin{enumerate}[(i)]
%(i) & j \in S_j,\\
\item $\lvert S\rvert \leq r+ \delta - 1$
\item For every $l \in S$,  $L = \{i_1,\ldots,i_{\lvert L \rvert}\} \subseteq S \setminus \{l\} $ and $\lvert L \rvert = \lvert S\rvert - (\delta - 1)$, $c_l$ is a function of $(c_{i_1}, \ldots, c_{i_{\lvert L \rvert}}))$, where ${\bf c}=(c_1,\ldots , c_n)\in C$.
\end{enumerate}
We say that $C$ is a  \emph{locally repairable code (LRC)}  with \emph{all-symbol locality} $(r,\delta)$ if all the $n$ symbols of the code are contained in an $(r,\delta)$-locality set. The locality sets can be also referred to as the local repair sets. 

We remark that the symbols in a locality set $S$ can be used to recover up to $\delta -1$ lost symbols in the same locality set.  Further, we note that each of the following statements are equivalent to statement (ii) above:
\begin{enumerate}
\item[(ii$'$)] For any
$l \in S$, $L= \{i_1,\ldots,i_{\lvert L \rvert}\} \subseteq S \setminus \{l\}$, and $\lvert L \rvert = \lvert S \rvert - (\delta - 1) $, we have 
$\lvert C_{L \cup \{l\}} \rvert = \lvert C_{L}\rvert$,

\item[(ii$''$)] For any $L \subseteq S$ with $|L| \geq |S| - (\delta - 1)$, we have $|C_L| = |C_S|$,
\item[(ii$'''$)] $d(C_{S}) \geq \delta$ , where $d(C_{S})$ is the minimum  distance of $C_{S}$.
%\end{array}
\end{enumerate}

An LRC with parameters $(n,k)$, minimum distance $d$, and all-symbol locality $(r,\delta)$ is an $(n,k,d,r,\delta)$-LRC. Since we focus only on all-symbol locality in this paper, we will henceforth use the term LRC to mean a locally repairable code with all-symbol locality. 

\subsection{The Singleton bound}\label{Sec:Singleton}
For any $\lbrack n,k \rbrack$-linear code with minimum distance $d$, the Singleton bound is given by
\begin{equation} \label{eq:singleton_bound}
d \leq n - k + 1.
\end{equation}
This bound was generalized for locally repairable codes in \cite{Gopalan} (the case $\delta=2$) and \cite{prakash12} (general $\delta$) as follows. A linear LRC with parameters $(n,k,d,r,\delta)$ satisfies
\begin{equation} \label{eq:singleton_bound_r_delta}
d \leq n - k + 1 - \left (\left \lceil \frac{k}{r} \right \rceil - 1 \right ) (\delta - 1).
\end{equation}
While the bounds in \cite{Gopalan} and \cite{prakash12} are stated assuming only information locality, so are of course in particular still valid under the stronger assumption of all-symbol locality. Other generalizations of the Singleton bound for linear and nonlinear LRCs can be found in \cite{LRCpapailiopoulos, cadambe13,  rawat14}.

\subsection{Graphs, $G=(V,E)$}
Let us fix some standard graph-theoretic notation that will be used at two stages in the constructions. A (finite) directed graph $G=(V,E)$ is a pair of a finite \emph{vertex set} $V$, whose elements are called nodes or vertices, and an \emph{edge set} $E\subseteq V\times V$ of pairs called arcs or edges. Graphs are often drawn with the vertices as points and arcs $(v,u)$ as arrows $v\to u$. We call $v$ the tail of $(v,u)$, and $u$ the head. A path from $S\subseteq V$ to $T\subseteq V$ is a sequence $v_0, v_1, \ldots, v_n$, where $v_0\in S$, $v_n\in T$, and $(v_i,v_{i+1})\in E$ for each  $i=0,\ldots ,n-1$.  If $v_0=v_n$, then the path is called a (directed) cycle. 

An important case of graphs is when $E$ is \emph{symmetric}, \emph{i.e.}, $(u,v)\in E$ if and only if $ (v,u)\in E$. In such case, it is customary to identify the two pairs $(u,v)$ and $(v,u)$ with the set $\{u,v\}$, and erase all the heads of the arrows in the drawing. When talking about a graph without specifying that it is directed, the symmetric situation is assumed. Observe that this definition allows for loops edges (where the tail and the head is the same), but not multiple edges.
In this paper, we will assume that all graphs, both symmetric and directed, are without multiple edges and loops.

\subsection{Posets and lattices, $(\mathcal{P},\subseteq)$}
Before studying matroids, we need a minimum of background on poset and lattice theory. We refer the reader to~\cite{stanley11} for more information on posets and lattices. The material in this section is used only in the technical work with the lattice of cyclic flats of matroids.

A collection of sets $\mathcal{P} \subseteq 2^E$ ordered by inclusion $\subseteq$ defines a (finite) poset $(\mathcal{P},\subseteq)$. A \emph{chain} $C$ of $(\mathcal{P},\subseteq)$ is a set of elements $X_0, \ldots,X_m \in \mathcal{P}$ such that $X_0 \subsetneq X_1 \subsetneq \ldots \subsetneq X_m$. The \emph{length} of a chain $C$ is defined as the integer $l(C) = \lvert C \rvert - 1 = m$. For $X,Y \in \mathcal{P}$, let 
$$
\begin{array}{l}
L_{X,Y} = \{Z \in \mathcal{P} : Z \subseteq X \hbox{ and } Z \subseteq Y\},\\
U_{X,Y} = \{Z \in \mathcal{P} : X \subseteq Z  \hbox{ and } Y \subseteq Z\}.
\end{array}
$$
An element $Z \in L_{X,Y}$ is the \emph{meet} of $X$ and $Y$, denoted by $X \wedge Y$, if it contains every $V \in L_{X,Y}$. Dually, $Z \in U_{X,Y}$ is the \emph{join} of $X$ and $Y$, denoted by $X \vee Y$, if it is contained in every $V \in U_{X,Y}$.  A poset $(\mathcal{P}, \subseteq)$ is a \emph{lattice} if every pair of elements of $\mathcal{P}$ has a meet and a join. If $(\mathcal{P},\subseteq)$ is a (finite) lattice, then there are two elements $0_\mathcal{P}, 1_\mathcal{P} \in \mathcal{P}$ such that $0_{\mathcal{P}} \subseteq X$ and $X \subseteq 1_\mathcal{P}$ for all $X \in \mathcal{P}$. The \emph{atoms} and \emph{coatoms} of a lattice $(\mathcal{L},\subseteq)$ are defined as
$$
\begin{array}{ll}
A_\mathcal{L} &= \{X \in \mathcal{L} \setminus 0_{\mathcal{L}} : \nexists Y \in \mathcal{L} \hbox{ such that } 0_\mathcal{L} \subsetneq Y \subsetneq X\},\\
coA_\mathcal{L} &= \{X \in \mathcal{L} \setminus 1_{\mathcal{L}} : \nexists Y \in \mathcal{L} \hbox{ such that } X \subsetneq Y \subsetneq 1_\mathcal{L}\},
\end{array}
$$
respectively.

\subsection{Matroids, $M = (\rho,E)$}
Matroids can be defined in many equivalent ways, for example by their rank function, nullity function, independent sets, circuits and more \cite{oxley92}. For our purpose, the following definition will be the most useful. Let $2^E$ denote the set of all subsets of $E$. A \emph{matroid} $M$ on a finite set $E$ is defined by a \emph{rank function} $\rho: 2^E \rightarrow \mathbb{Z}$ satisfying the following axioms:
\begin{equation} \label{eq:rank_matroid}
\begin{array}{rl}
(R1) & 0 \leq \rho(X) \leq \lvert X \rvert \hbox{ for } X \subseteq E,\\
(R2) & X \subseteq Y \subseteq E \Rightarrow  \rho(X) \leq \rho(Y),\\
(R3) & X, Y \subseteq E \Rightarrow \rho(X) + \rho(Y) \geq \rho(X \cup Y) + \rho(X \cap Y).
\end{array}
\end{equation}
The \emph{nullity function} $\eta:2^E \rightarrow \mathbb{Z}$ of the matroid $M=(E,\rho)$ is defined by
$$
\eta(X) = \lvert X \rvert - \rho(X), \hbox{ for } X \subseteq E.
$$
Let $X$ be any subset of $E$. The subset $X$ is \emph{independent} if $\rho(X) = \lvert X \rvert$, otherwise it is \emph{dependent}. A dependent set $X$ is a $\emph{circuit}$ if all proper subsets of $X$ are independent, \emph{i.e.}, $\rho(X) = \lvert X \rvert - 1$ and $\rho(Y) = \lvert Y \rvert$ for all subsets $Y \subsetneq X$. The \emph{closure} of $X$ is defined as
$$
\mathrm{cl}(X) = \{x \in E : \rho(X \cup x) = \rho(X)\}. 
$$
The subset $X$ is a \emph{flat} if $\mathrm{cl}(X)= X$. It is \emph{cyclic} if it is a (possible empty) union of circuits. 
%A flat that is cyclic is called a \emph{cyclic flat}. 
The sets of circuits, independent sets,  cyclic sets and cyclic flats of a matroid $M$ is denoted by $\mathcal{C}(M)$, $\mathcal{I}(M)$, $\mathcal{U}(M)$ and $\mathcal{Z}(M)$, respectively. We omit the subscript $M$ when the matroid is clear and write $\mathcal{C}$, $\mathcal{I}$, $\mathcal{U}$ and $\mathcal{Z}$, respectively. The set of cyclic flats together with inclusion defines the \emph{lattice of cyclic flats} $(\mathcal{Z},\subseteq)$ of the matroid. The \emph{restriction} of $M$ to $X$ is the matroid $M \vert X = (\rho_{\vert X}, X)$ where
\begin{equation} \label{eq:rank_restriction}
\rho_{\vert X}(Y) = \rho(Y) \hbox{, for all subsets } Y \subseteq X.
\end{equation}
%The \emph{dual} of the matroid $M = (\rho,E)$ is the matroid $M^* = (\rho^*,E)$, where
%\begin{equation} \label{eq:dual-rank}
%\rho^*(X) = \rho(E \setminus X) + \lvert X \rvert - \rho(E), \hbox{ for } X \subseteq E.
%\end{equation}
%If we omit the argument $M^*$, then we denote the sets of circuits and cyclic flats of the dual matroid by $\mathcal{C}^*$ and $\mathcal{Z}^*$, respectively.

\subsection{Almost affine codes and their associated matroids}~\label{Sec:AlmostAffine}

A code $C \subseteq \mathbb{A}^n$, where $\mathbb{A}$ is a finite set of size $s \geq 2$, is \emph{almost affine} if 
$$
\log_{s}(\lvert C_X \rvert) \in \mathbb{Z}
$$
for each $X \subseteq \lbrack n \rbrack$. Note that if $C$ is an almost affine code, then all projections $C_X$ of $C$ are also almost affine. 

%Two special cases of almost affine codes are linear codes and $m$-multilinear codes over finite fields. A code $C \subseteq \mathbb{F}_q^{mn}$ is an \emph{m-multilinear code} over $\mathbb{F}_q$ if $C$ is a linear code over $\mathbb{F}_q $ and, when considering $C$ as a code in $(\mathbb{F}_q^m)^n$, the dimension of the vector space $C_X$ over $\mathbb{F}_q$ is divisible by $m$ for all $X \subseteq \lbrack n \rbrack$. We remark that a linear code is a 1-multilinear code and that an $m$-multilinear code is a special case of vector-linear codes (linear codes where each code symbol is a vector). Not all vector-linear codes are almost affine and vice versa. 

In \cite{simonis98} it is proven that every almost affine code 

$C \subseteq \mathbb{A}^n$ 

induces a matroid $M_C = (\rho_C, \lbrack n \rbrack)$, where 
\begin{equation} \label{eq:rank_almost_affine}
\rho_C(X) = \log_s(\lvert C_X \rvert).
\end{equation}
Examples of matroids which cannot be represented by any almost affine code are given in 
%FIX NUMBER
\cite{matus99}. Moreover, an example of a matroid which can be represented by an almost affine code over a three letter alphabet, but not by any linear code is given in \cite{simonis98}. This example is the so-called non-Pappus matroid.

\begin{exam} \label{exam:matroid}
An example of a matroid $M_G = (\rho,E)$ is defined by the matrix \begin{equation}\label{example}G=
\bordermatrix{
 &{\bf 1}& {\bf 2}&{\bf 3}&{\bf 4} &{\bf 5}&{\bf 6}&{\bf 7}& {\bf 8}& {\bf 9}&{\bf 10}&{\bf 11}&{\bf 12} \cr  
a&1& & & & & &1& &1&3& &1 \cr
b& &1& & & & &1& & &2& &  \cr
c& & &1& & & &1&1& &1&3&  \cr
d& & & &1& & & &1& & &2&  \cr
e& & & & &1& & &1&1& &1&3 \cr
f& & & & & &1& & &1& & &2
},
\end{equation} which we think of as a generator matrix of a linear code $C$ over the field $\mathbb{F}_5$. The code $C$ is the row span of $G$, $E=\{{\bf 1},\ldots,{\bf 12}\}$ is the set of columns, and the rank of a subset of $E$ is the rank of the corresponding submatrix, \emph{i.e.}, 
$$
\rho(I) = rank(G_I) \hbox{ for } I \subseteq E,
$$
where $G_I$ is the submatrix of $G$ whose columns are the columns indexed by $I$. Below are some independent sets, circuits, cyclic flats and rank functions of some subsets of $E$ for the matroid $M$.
$$
\begin{array}{l}
\mathcal{I} = \{\emptyset, \{2,3,7\}, \{3,4,5\}, \{7,8,9\}, [6],  \ldots\},\\
\mathcal{C} = \{\{1,2,3,7\}, \{4,5,8,11\}, \ldots\},\\
\mathcal{Z} = \{\{1,2,3,7,10\}, \{3,4,5,8,11\}, \{1,2,3,4,5,7,8,10,11\},[12], ...\},\\
\rho(\emptyset) = 0 \hbox{, }\rho(\{3,4,5\}) = \rho(\{4,5,8,11\}) = \rho(\{3,4,5,8,11\}) = 3 \hbox{, } \rho([6]) = \rho([12]) = 6.
\end{array}
$$ 
The reader can verify that the code generated by this matrix corresponds to the storage system in Figure~\ref{fig:cloud}, when the rows are the information symbols.
\end{exam}

%\begin{exam}
%In \cite[Ex. 10]{rouayheb10} we have the following 2-multilinear representation over $\mathbb{F}_3$ of the non-Pappus matroid. The non-Pappus matroid is an example of a matroid which cannot be represented by any linear code $C$ as $M = M_C$. For the matrix
%$$
%G = 
%\left (
%\begin{array} {ccccccccc}
%10 & 10 & 00 & 10 & 00 & 10 & 10 & 10 & 00\\
%01 & 01 & 00 & 01 & 00 & 01 & 01 & 01 & 00\\
%00 & 00 & 00 & 10 & 10 & 21 & 01 & 10 & 10\\
%00 & 00 & 00 & 02 & 01 & 20 & 12 & 02 & 01\\
%00 & 10 & 10 & 01 & 00 & 01 & 00 & 11 & 10\\
%00 & 01 & 01 & 21 & 00 & 21 & 00 & 10 & 01
%\end{array}
%\right ) .
%$$
%over $\mathbb{F}_3^2$, let $C = \{\boldsymbol{x}G : \boldsymbol{x} \in \mathbb{F}_3^6\} \subseteq (\mathbb{F}_3^2)^9$. The code $C$ induces a matroid $M_C$ on $\lbrack 9 \rbrack$ with the following rank function
%$$
%\rho_C(X) =
%\left \{
%\begin{array}{lcl}
%2 & \hbox{if} & X \in Y,\\
%\min(\lvert X \rvert, 3) & \hbox{if} & X \in 2^{[9]} \setminus Y,
%\end{array}
%\right .
%$$
%where 
%$$
%\begin{array}{rcl}
%Y & = & \{ \{1,2,3\},\{1,5,7\},\{1,6,8\},\{2,4,7\},\\
%  &   & \hbox{ }\{2,6,9\},\{3,4,8\},\{3,5,9\},\{4,5,6\}\}.
%\end{array}
%$$ 
%Given below is the lattice of cyclic flats of $M_C$ and the ranks of the cyclic flats denoted by {\color{blue}???}

%{\color{blue}INSERT A PICTURE}\\ 
%From the lattice of cyclic flat above and Lemma \ref{lemma:(n,k,d,r)-Z} we get that $n = 9$, $k = 3$, $d = n - k + 1 - (3 - 2) = 6$ and $r = 2$. Since, $n - k - \lceil \frac{k}{r} \rceil + 2 = 6$ we get that $C$ is an %optimal almost affine LRC.
%\end{exam}

\subsection{Basic properties of matroids and the lattice of cyclic flats}

For the applications in this paper, the most important matroid attribute is its lattice of cyclic flats.  This is because the minimal cyclic flats of matroids will correspond to local repair sets of the LRC. In this subsection, we present basic properties of the lattice of cyclic flats, that will be needed in later parts of the paper. 

\begin{proposition} [see \cite{bonin08}] \label{pro:basic-Z}
Let $M = (\rho,E)$ be a matroid. Then

\begin{enumerate}[(i)]
\item $\rho(X) = \min \{\rho(F) + \lvert X \setminus F \rvert : F \in \mathcal{Z}\}$, for $X \subseteq E$,
\item Define $ \mathcal{D} = \{ X: \mbox{there is } F \in \mathcal{Z} \mbox{ with } X\subseteq F \mbox{ and }  \lvert X \rvert = \rho(F) + 1 \}$. \\
Then $\mathcal{C}$ is the set of minimal elements in $\mathcal{D}$, ordered by inclusion.
% = \{X \in \mathcal{D} : \nexists Y \in  \mathcal{D} \hbox{ such that } Y \subsetneq X\}$,
\item $(\mathcal{Z}, \subseteq)$ is a lattice with the following meet and join for $X,Y \in \mathcal{Z}$,\\
 $X \wedge Y = \bigcup_{\{C \in \mathcal{C} : C \subseteq X \cap Y\}} C$ and $X \vee Y = \mathrm{cl}(X \cup Y)$.
\end{enumerate}

\end{proposition}

The assertion (i) in Proposition \ref{pro:basic-Z} shows that a matroid is determined by its cyclic flats and their ranks. Conversely, the following theorem gives an axiomatic scheme for a collection of subsets on $E$ and a function on these sets to define the cyclic flats of a matroid and their ranks.  This will allow us to construct matroids with prescribed parameters in Section~\ref{sec:(n,k,d,r,delta)-matroid}.

\begin{thm} [see \cite{bonin08} Th. 3.2] \label{th:Z-axiom}
Let $\mathcal{Z} \subseteq 2^E$ and let $\rho$ be a function $\rho: \mathcal{Z} \rightarrow \mathbb{Z}$. There is a matroid $M$ on $E$ for which $\mathcal{Z}$ is the set of cyclic flats and $\rho$ is the rank function restricted to the sets in $\mathcal{Z}$ if and only if
$$
\begin{array}{rl}
(Z0) & \mathcal{Z} \hbox{ is a lattice under inclusion},\\
(Z1) & \rho(0_{\mathcal{Z}}) = 0,\\
(Z2) & X,Y \in \mathcal{Z} \hbox{ and } X \subsetneq Y \Rightarrow \\
     & 0 < \rho(Y) - \rho(X) < \lvert Y \rvert - \lvert X \rvert,\\
(Z3) & X,Y \in \mathcal{Z} \Rightarrow \rho(X) + \rho(Y) \geq \\
     & \rho(X \vee Y) + \rho(X \wedge Y) + \lvert (X \cap Y) \setminus (X \wedge Y)  \rvert.
\end{array}
$$
\end{thm}

 The results in the proposition below are basic matroid results that will be needed several times in the proofs of other results given later in this paper. We give a proof for the results in  Proposition \ref{prop:basic_facts} that we have not been able to find in the literature. For the other results we only give a reference. 

\begin{proposition} \label{prop:basic_facts}
Let $M = (\rho,E)$ be a matroid and let $X,Y$ be subsets of $E$, then
\begin{enumerate}[(i)]
\item If $X \subseteq Y$, then $\eta(X) \leq \eta(Y)$,
\item $\eta(X \cup Y) \geq \eta(X) + \eta(Y) - \eta(X \cap Y)$,
\item If $\rho(X) < \rho(E)$ and $1_{\mathcal{Z}} = E$, then $\eta (X) \leq  \max \{ \eta (Z) : Z \in coA_{\mathcal{Z}} \}$,
\item $\mathrm{cl}(U) \in \mathcal{Z}(M)$ for $U \in \mathcal{U}(M)$,
\item $\mathcal{U}(M \vert X) = \{U \subseteq X : U \in \mathcal{U}(M)\}$,
\item $\mathcal{C}(M \vert X) = \{C \subseteq X : C \in \mathcal{C}(M)\}$,
\item $\mathcal{Z}(M \vert X) = \{Z \in \mathcal{Z}(M) : Z \subseteq X\}$ if $X \in \mathcal{F}(M)$,
\item  $X \notin \mathcal{U}(M)$ if and only if $\exists x \in X$ such that $\rho(X-x) < \rho(X)$, 
\item $\rho(\mathrm{cl}(X)) = \rho(X)$,
\item If $X \subseteq Y$, then $\mathrm{cl}(X) \subseteq \mathrm{cl}(Y)$.
\end{enumerate}
\end{proposition} 

\begin{proof}
 Properties (i), (ii), (v), (vii) and (viii) can be found in \cite[Lemma 2.2.4, Lemma 2.3.1, the paragraph under Lemma 2.4.5]{shoda12}.  Property (iv) is a consequence of \cite[Proposition 1.4.10 (ii)]{oxley92}. For (iii), assume that $\rho(X) < \rho(E)$ and $1_\mathcal{Z} = E$. Thus, $\mathrm{cl}(X) \neq E$ and $\eta(X) \leq \eta(\mathrm{cl}(X))$. Let $U$ be the largest cyclic set such that $U \subseteq \mathrm{cl}(X)$. From \cite[Lemma 2.4.8, Lemma 2.5.2]{shoda12}, we have that  $\eta(\mathrm{cl}(X)) = \eta(U)$ and that $U$ is a cyclic flat. Property (iv) now follows from the fact that 
$$
\rho(U) \leq \rho(\mathrm{cl}(X)) < \rho(E) = \rho(1_\mathcal{Z}).
$$ 
Property (vi) is a direct consequence of (v). Property (ix) is a consequence of property (x) which can be found in \cite[Lemma 1.4.2]{oxley92} 
\end{proof}

\begin{exam} \label{exam:Z}
Continuing with Example \ref{exam:matroid}, and remembering that the elements of $M_G$ are the columns of $G$, we see that the cyclic flats of $M_G$ are the submatrices in Figure~\ref{fig:cyclic}. The atomic cyclic flats are thus the submatrices corresponding to column sets $\{1,2,3,7,10\}$, $\{3,4,5,8,11\}$ and $\{1,5,6,9,12\}$. Remembering from~\eqref{example} that the rows are indexed by the information symbols $(a,b,c,d,e,f)$, these atomic cyclic flats agree exactly with the local clouds in Figure~\ref{fig:cloud}.

\begin{figure}[!htb]
    \centering
    \includegraphics[height = 9cm, width=0.68\textwidth]{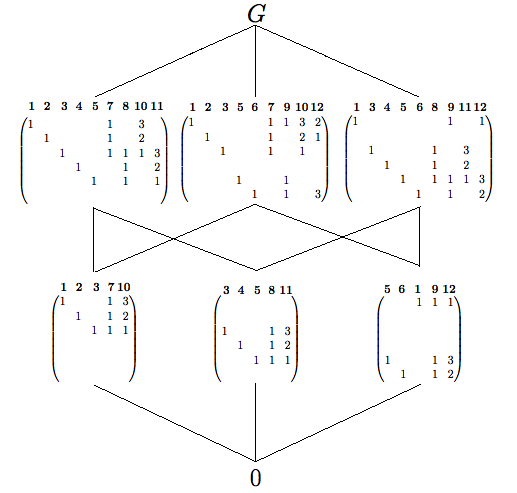}
    \caption{The lattice $\mathcal{Z}(M_G)$ of cyclic flats of the matroid $M(G)$ in Example~\ref{exam:Z}.}
    \label{fig:cyclic}
\end{figure}
\end{exam}

\section{Locally Repairable Matroids} \label{sec:(n,k,d,r,delta)-matroid}

\subsection{The parameters $(n,k,d,r,\delta)$ for matroids} \label{sec:(n,k,d,r,delta)-matroids}

In this subsection we show that the parameters $(n,k,d,r,\delta)$ are matroid invariants for an almost affine LRC. This will allow us to extend the definition of these parameters to matroids in general. 

 Let $C$ be an almost affine $(n,k,d,r,\delta)$-LRC over some finite alphabet $\mathbb{A}$. By the definition given in Eq. \eqref{eq:rank_almost_affine}, we know that $
\lvert C_X \rvert = \lvert \mathbb{A} \rvert^{\rho_C(X)}$, which specializes to $k = \rho_C(\lbrack n \rbrack)$ when $X=[n]$. 
In \cite{simonis98} it is proven that $M_{C_X} = M_C \vert X$ 
for $X \subseteq \lbrack n \rbrack$. Consequently, since the projection $C_X$ is also almost affine,  \eqref{eq:d_via_projection} implies that 
$$
d(C_X) = \min \{|Y | : Y \subseteq X \hbox{ and } \rho_C(X \setminus Y) < \rho_C(X) \},
$$
where $d(C_X)$ denotes the minimum distance of $C_X$. 

Using the observations above and the definition of an $(n,k,d,r,\delta)$-LRC  given in Section \ref{sec:LRC}, we conclude the following theorem. 

\begin{thm} \label{thm:matroid_invariant}
Let $C$ be an almost affine LRC with the associated matroid $M_C = (\rho_C, \lbrack n \rbrack)$. Then, the parameters $(n,k,d,r,\delta)$ of $C$ are matroid invariants, where
\begin{enumerate}[(i)]
\item $k = \rho_C(\lbrack n \rbrack)$,
 \item $d = \min \{ \lvert X \rvert : X \subseteq [n] \hbox{ and } \rho_C([n] \setminus X) < k\}$,
\item $C$ has all-symbol locality $(r,\delta)$ if and only if, for every $j \in \lbrack n \rbrack$  there exists a subset $S_j \subseteq \lbrack n \rbrack$ such that 
\begin{enumerate}
\item $j \in S_j$,
\item $\lvert S_j \rvert \leq r + \delta - 1$,
\item $d(C_{S_j}) = \min \{ \lvert X \rvert : X \subseteq S_j \hbox{ and } \rho_C(X) < \rho_C(S_j) \} \geq \delta$.
\end{enumerate} 
\end{enumerate}
\end{thm}

These results can now be taken as the definition of the parameters $(n,k,d,r,\delta)$ for an arbitrary matroid.
 \begin{definition} \label{def:(n,k,d,r,delta)-matroid}
Let $M = (\rho,E)$ be a matroid. Then we call $M$ an $(n,k,d,r,\delta)$-matroid, where
\begin{enumerate}[(i)]
\item $n = \lvert E \rvert$,
\item $k = \rho(E)$,
\item $d = \min \{ | X | : X \subseteq E \hbox{ and } \rho(E \setminus X) < k \}$,
\item %for  $0 < r \leq \rho(E)$ and $\delta \geq 2$, $M$  has all-symbol locality $(r,\delta)$ if and only if f
The parameters $0<r\leq \rho(E)$ and $\delta\geq 2$ are such that for all $x \in  E$, there exists a subset $S_x \subseteq E$ with
\begin{enumerate}
\item $x \in S_x$,
\item $\lvert S_x \rvert \leq r + \delta - 1$,
\item $d(M|S_x) = \min \{ \lvert X \rvert : X \subseteq S_x \hbox{ and } \rho(S_x \setminus X) < \rho(S_x) \} \geq \delta$.
\end{enumerate} 
\end{enumerate}  
\end{definition} 

 A subset $S \subseteq E$ is called a $(r,\delta)$-\emph{locality set} of the elements $x \in S$ if the statements b)--c)  above are satisfied by $S$. The parameters $n$ and $k$ are obviously defined for all matroids. We note that the parameter $d$ is finite if and only $k > 0$. Furthermore, we notice that every element $x \in E$ is contained in some cyclic set $S_x$ if and only if $1_\mathcal{Z} = E$. If this is the case, and $r = \max \{\lvert S_x \rvert - 1 : x \in X\}$, then $M$ has $(r,2)$-locality. 
As a consequence of the observations above, we get the following proposition. 

\begin{proposition} \label{proposition:parameters_defined}
A matroid $M=(\rho,E)$ is an $(n,k,d,r,\delta)$-matroid with finite values of $(n,k,d,r,\delta)$ if and only if $0 < \rho(E)$ and $1_{\mathcal{Z}} = E$.
\end{proposition} 

% Hence, by an \emph{$(n,k,d,r,\delta)$-matroid} $M = (\rho,E)$ we mean an $(n,k,d)$-matroid with $k > 0$, $1_{\mathcal{Z}} = E$ and all-symbol locality $(r,\delta)$.  Further, we o
Observe that if $M$ has $(r,\delta)$-locality, then  by Definition \ref{def:(n,k,d,r,delta)-matroid} (iv), $M$ has $(r',\delta')$-locality for $r \leq r' \leq k$ and $2 \leq \delta' \leq \delta$ with $r' + \delta' \geq r+\delta$. So neither the values of $(r,\delta)$ nor the locality sets $S_x$ are in general uniquely determined for a 

matroid

 $M$.
%\begin{proposition} \label{prop:conditions_(n,k,d,r,delta)}
%Let $M = (\rho,E)$ be a matroid. If the parameters $(n,k,d,r,\delta)$ is defined for $M$, then $0 < r \leq k \leq n - (\delta - 1)\lceil \frac{k}{r} \rceil$, $1_{\mathcal{Z}} = E$ and $2 \leq \delta \leq d$. The parameters $(n,k,d,r,\delta)$ for $M$ is defined if $0 < \rho(E)$ and $1_{\mathcal{Z}} = E$. If $r' = \min\{r : M \hbox{ has } (r,\delta)\hbox{-locality}\}$ \hbox{ and } $\delta' = \max \{\delta : M \hbox{ has } (r,\delta) \hbox{-locality}\}$ \hbox{, then $M$ has $(r',\delta')$-locality}. 
%\end{proposition} 

\subsection{A generalized Singleton bound for $(n,k,d,r, \delta)$-matroids}

%In this subsection we will first relate the parameters $(n,k,d,r,\delta)$ of a matroid to its lattice of cyclic flats, in Lemmas~\ref{lemma:(n,k,d,r,delta)-Z}--\ref{lemma:fundamental_chain_bound_d} below. Using these lemmas, we prove Theorem~\ref{th:bound}, which is a generalized Singleton bound for matroids. 

The main result of this subsection is Theorem \ref{th:bound} which gives a Singleton-type bound on the parameters $(n,k,d,r,\delta)$ for matroids. 
%The corresponding bound \eqref{eq:singleton_bound_r} for L
In the case of linear LRCs with information locality and trivial failure tolerance $\delta =2$, \emph{i.e.}, only tolerating one failure, the bound was given in \cite{Gopalan}. 
%To prove this bound in \cite{Gopalan}, the authors show the existence of a large enough set $S$ with rank less that $k-1$. This set $S$ is constructed by an algorithm which uses some smaller sets associated to the informations symbols. The bound now follows by the use of $S$ and a result on how $d$ corresponds to the largest size of a set with rank $k-1$. 

The core ingredients of the proof of Theorem \ref{th:bound} are the same as in \cite{Gopalan}, interpreted for matroids. First, we relate the parameters $(n,k,d,r,\delta)$ of a matroid to its lattice of cyclic flats in Theorem \ref{theorem:(n,k,d,r,delta)-Z}. Then in Lemma \ref{lemma:fundamental_chain},  we obtain a large cyclic flat $Y_{m-1}$ of rank less than $k$.  In Theorem \ref{th:bound} we relate $Y_{m-1}$ to $d$, thereby proving the theorem.

\begin{thm} \label{theorem:(n,k,d,r,delta)-Z} 
Let $M = (\rho,E)$ be an $(n,k,d,r,\delta)$-matroid with $0 < \rho(E)$ and $1_\mathcal{Z} = E$. Then 
\begin{enumerate}[(i)]
\item $d = n - k + 1 - \max \{ \eta(Z) : Z \in coA_{\mathcal{Z}} \}$,
\item For each $x \in E$, there is a cyclic set
$S_x \in \mathcal{U}(M)$ such that
\begin{enumerate}
\item $x \in S_x$,
\item $\lvert S_x \rvert \leq r + \delta - 1$,
\item $d(M \vert S_x) = \eta(S_x) + 1 - \mathrm{max} \{ \eta(Z) : Z \in coA_{\mathcal{Z}(M \vert S_x)}\} \geq \delta$.
\end{enumerate}
\end{enumerate}
\end{thm}

\begin{proof}
The proof is given in the Appendix.
\end{proof}

As $\eta(Z)$ is non-negative for every $Z$, Theorem~\ref{theorem:(n,k,d,r,delta)-Z}~(ii)~c) gives $\delta + \rho(S_x) - 1 \leq \lvert S_x\rvert$, which together with Theorem~\ref{theorem:(n,k,d,r,delta)-Z}~(ii)~b) shows that
\begin{equation} \label{eq:rank_S_x}
\rho(S_x) \leq r
\end{equation}
for any $(r,\delta)$-locality $S_x$. Moreover, we observe that for any atom $S$ in a lattice of cyclic flats with $0_\mathcal{Z} = \emptyset$, we can use any subset $S' \subseteq S$ as a locality set when $|S'| > \rho(S)$. However, different choices of locality sets may give different values on the parameters $(r,\delta)$.

\begin{exam} \label{exam:Z_(n,k,d,r,delta)}
Representing the cyclic flats associated to the matroid $M_G$ from Example \ref{exam:Z} just by their corresponding sets and ranks in Figure~\ref{fig:cyclicsets},
we use Theorem \ref{theorem:(n,k,d,r,delta)-Z} to get the parameters $(n,k,d,r,\delta)$ of the linear LRC that is generated by the matrix $G$ given in Example \ref{exam:matroid}.  

\begin{figure}[!htb]
    \centering
    \includegraphics[height = 5cm, width=0.5\textwidth]{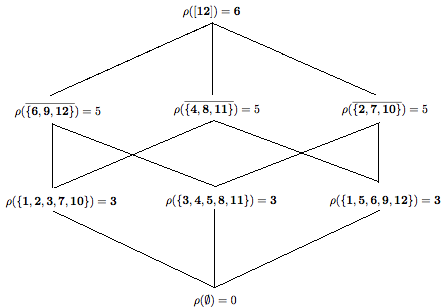}
    \caption{The lattice $\mathcal{Z}(M_G)$ of cyclic flats of the matroid $M(G)$ in Example~\ref{exam:Z}, without reference to the matrix $G$.}
    \label{fig:cyclicsets}
\end{figure}

The values for $(n,k,d)$ are
$$
\begin{array}{l}
n = 12,\\
k = 6,\\
d = 12 - 6 + 1 - 4 = 3.
\end{array}
$$
Using $S_1 = \{1,2,3,7,10\}$, $S_2 = \{3,4,5,8,11\}$ and $S_3 = \{1,5,6,9,12\}$ as the locality sets, we get the parameters $(r,\delta) = (3,3)$. 
%If we choose some other locality sets, e.g. $S'_1 = \{2,3,7,10\}$, $S'_2 = \{4,5,8,11\}$, $S'_3 = \{1,5,6,9\}$ and $S'_4 = \{5,6,9,12\}$, we get the parameters $(r,\delta) = (3,2)$.
\end{exam}

From Theorem \ref{theorem:(n,k,d,r,delta)-Z}, we derive a chain of cyclic flats, from which we will extract a large cyclic flat, to be used in the proof of Theorem~\ref{th:bound}.

 \begin{lemma} \label{lemma:fundamental_chain}
Let $M = (\rho,E)$ be an $(n,k,d,r,\delta)$-matroid. %and let  $\{S_x\}_{x \in E}$ be a collection of cyclic sets of $M$ for which the statements $a) - c)$ in Theorem \ref{theorem:(n,k,d,r,delta)-Z}~(ii) are satisfied. 
Then there is a chain $$0_\mathcal{Z} = Y_0 \subsetneq Y_1 \subsetneq \ldots \subsetneq Y_m = E$$ in $(\mathcal{Z}(M),\subseteq)$  such that for $j = 1,\ldots,m$ we have
\begin{enumerate}[(i)]
%\item $C: 0_\mathcal{Z} = Y_0 \subsetneq Y_1 \subsetneq \ldots \subsetneq Y_m = E$ is a chain in $(\mathcal{Z}(M),\subseteq)$, 
\item $\rho(Y_j) %\leq \rho(Y_{j-1}) + \rho(S_j) - \rho(Y_{j-1} \cap S_j) 
\leq \rho(Y_{j-1}) + r$,
\item $\eta(Y_j) %\geq \eta(Y_{j-1}) + \eta(S_j) - \eta(Y_{j-1} \cap S_j) 
\geq \eta(Y_{j-1}) + (\delta - 1)$.
\end{enumerate}
\end{lemma} 

\begin{proof}
The proof is given in the Appendix.
\end{proof}

%The following lemma is now a consequence of Theorem \ref{theorem:(n,k,d,r,delta)-Z} and Lemma \ref{lemma:fundamental_chain}.

%\begin{lemma} \label{lemma:fundamental_chain_bound_d}
%Let $M = (\rho,E)$ be an $(n,k,d,r,\delta)$-matroid with parameters $(n,k,d,r,\delta)$, and let  
%$$
%C: 0_\mathcal{Z} = Y_0 \subsetneq Y_1 \subsetneq \ldots \subsetneq Y_m = E
%$$
%be any chain of $(\mathcal{Z},\subseteq)$ given in Lemma \ref{lemma:fundamental_chain}. Then
%$$
%d \leq n - k + 1 - \eta(Y_{m-1}) \quad \hbox{and} \quad m \geq \left \lceil \frac{k}{r} \right \rceil.
%$$
%\end{lemma}

%\begin{proof}
%The proof is given in the Appendix. 
%\end{proof}

%The following generalized Singleton bound for matroids, now follows almost immediately from Lemmas~\ref{lemma:fundamental_chain}--\ref{lemma:fundamental_chain_bound_d}.
We are now ready to prove the generalized Singleton bound for matroids.

\begin{thm} \label{th:bound}
Let $M= (\rho,E)$ be an $(n,k,d,r,\delta)$-matroid. Then
$$
d \leq n - k + 1 - \left (\left \lceil \frac{k}{r} \right \rceil - 1 \right )(\delta - 1).
$$
\end{thm}

\begin{proof}
 Let 
$$
C: 0_{\mathcal{Z}} = Y_0 \subsetneq Y_1 \subsetneq \ldots \subsetneq Y_m = E
$$
be a chain of $(\mathcal{Z},\subseteq)$ given in Lemma \ref{lemma:fundamental_chain}. 
Then $\eta(Y_{m-1}) \geq (m-1) (\delta - 1)$, by Lemma \ref{lemma:fundamental_chain} (ii). On the other hand, by Lemma \ref{lemma:fundamental_chain} (i) we have that $k=\rho(Y_m)\leq mr$, so $m\geq \lceil\frac{k}{r}\rceil$.

Combining these results we get $$\eta(Y_{m-1}) \geq (m-1) (\delta - 1)\geq (\left\lceil\frac{k}{r}\right\rceil-1)(\delta-1).$$

Since $Y_{m-1}\in\mathcal{Z}\setminus \{1_\mathcal{Z}\}$, we have 
$$
\max \{ \eta(Z) : \eta(Z) \in coA_\mathcal{Z} \}  \geq  \eta(Y_{m-1}),
$$
so Theorem~\ref{theorem:(n,k,d,r,delta)-Z}~(i) yields
$$
d   \leq  n - k + 1 - \eta(Y_{m-1})\leq n - k + 1 - \left (\left \lceil \frac{k}{r} \right \rceil - 1 \right )(\delta - 1).
$$
\end{proof}

We also give three additional bounds on the parameters of a matroid.

\begin{proposition} \label{proposition:bound_delta_k_rate}
Let $M = (\rho,E)$ be an $(n,k,d,r,\delta)$-matroid. Then 
\begin{enumerate}[(i)]
\item $\delta \leq d$,
\item $k \leq n - \left \lceil \frac{k}{r} \right \rceil (\delta - 1)$,
\item $\frac{k}{n} \leq \frac{r}{r + \delta - 1}$.
\end{enumerate}
\end{proposition}

\begin{proof}
The proof is given in the appendix.
\end{proof}

In the case of codes, Proposition~\ref{proposition:bound_delta_k_rate} (i) and (iii) have natural interpretations. Indeed, (i) says that the local minimum distance is bounded from above by the global minimum distance, and (iii) says that the global code rate is bounded from above by the local code rate.

\subsection{A structure theorem for matroids achieving the generalized Singleton bound}

\begin{definition}\label{Def:perfect} We will call an $(n,k,d,r,\delta)$-matroid \emph{perfect} if it meets the generalized Singleton bound of Theorem \ref{th:bound} with equality, \emph{i.e.} if \begin{equation}\label{perfect}
d = n-k + 1 -\left ( \left \lceil \frac{k}{r} \right \rceil - 1 \right ) (\delta - 1).
\end{equation} In analogy, we will call a LRC satisfying \eqref{perfect} a \emph{perfect} LRC\footnote{We point out that, typically, codes achieving these kind of bounds have been called optimal in the literature. However, we feel that the notion \emph{optimal} should be saved for the code that is the best we can do. Thus, saying that an optimal code does not exist when the bound cannot be reached with equality feels wrong, since we can still find a code with minimum distance only slightly smaller than the bound, and this code is the best possible solution in this case and thus deserves to be called optimal. Therefore, we have opted to call the codes achieving the bound \emph{perfect}. This is say that, even though perfect codes do not exist for all parameters, optimal solutions can still be found.}. 
\end{definition}

These notions should not be confused with those of a perfect matroid design or a perfect code in classical coding theory literature.  
Theorem \ref{th:structure-optimal-matroid} gives some necessary structural properties for perfect $(n,k,r,\delta)$-matroids with $r < k$.  We will use this structure theorem to prove that for certain values of $(n,k,r,\delta)$, there are no perfect $(n,k,r,\delta)$-matroids, and consequently no perfect LRCs.
 The degenerate case when $r=k$ is easier, and is considered in Section \ref{sec:max_d_r=k}.

A collection of sets $X_1 \ldots, X_j$ is said to have a \emph{non trivial union} if 
$$
X_l \nsubseteq \bigcup_{i \in \lbrack j \rbrack \setminus \{l\}} X_i \hbox{, for } l = 1,\ldots,j.
$$

\begin{thm} \label{th:structure-optimal-matroid}
Let $M=(\rho,E)$ be an $(n,k,d,r,\delta)$-matroid with $r < k$ and
$$
d = n-k + 1 -\left ( \left \lceil \frac{k}{r} \right \rceil - 1 \right ) (\delta - 1).
$$
Let then $\{S_x : x \in E\} \subseteq \mathcal{U}(M)$ be a collection of cyclic sets for which the statements a) -- c) in Theorem \ref{theorem:(n,k,d,r,delta)-Z}~(iv) are satisfied. Then 
\begin{enumerate}[(i)]
\item $0_\mathcal{Z} = \emptyset$,
\item for each $x \in E$,
	\begin{enumerate}
	\item[$a$)] $\eta(S_x) = (\delta - 1)$,
	\item [$b$)]$S_x$ is an atom in $ \mathcal{Z}(M)$, and in particular a cyclic flat. 
	\end{enumerate}
\item For each collection $F_1, \ldots F_j$ of cyclic flats in $\{S_x : x \in E\}$ that has a non trivial union,
		 \begin{enumerate}
		 \item[$c$)]$\eta(\bigvee_{i=1}^j F_i) = 
\left \{
\begin{array}{lcl}
j(\delta - 1) & \hbox{if} & j <  \lceil \frac{k}{r} \rceil,\\
n-k \geq  \lceil \frac{k}{r} \rceil (\delta - 1) & \hbox{if} & j \geq 	 \lceil \frac{k}{r}  \rceil,
\end{array}
\right .$
		\item[$d$)] $\bigvee_{i=1}^j F_i = 
\left \{
\begin{array}{lcl}
\bigcup_{i = 1}^j F_i & \hbox{if} & j <  \lceil \frac{k}{r} \rceil,\\
E & \hbox{if} & j \geq 	 \lceil \frac{k}{r}  \rceil,
\end{array}
\right .$
	\item[$e$)] $\rho(\bigvee_{i=1}^j F_i) = 
\left \{
\begin{array}{lcl}
\lvert \bigcup_{i = 1}^j F_i \rvert - j(\delta - 1) & \hbox{if} & j <  \lceil \frac{k}{r} \rceil,\\
k & \hbox{if} & j \geq 	 \lceil \frac{k}{r}  \rceil.
\end{array}
\right .$
	\item[$f$)] $|F_j \cap (\bigcup_{i = 1}^{j-1} F_i) | \leq |F_j| - \delta$ if $j \leq  \lceil \frac{k}{r}  \rceil$.
\end{enumerate}
\end{enumerate} 
\end{thm}

\begin{proof}
The proof is given in the Appendix.
\end{proof}

By the structure theorem \ref{th:structure-optimal-matroid} above we get the following corollary.

\begin{corollary} \label{corollary:structure-optimal-matroid}
Let $M=(\rho,E)$ be an $(n,k,d,r,\delta)$-matroid with $r < k$ and
$$
d = n-k + 1 -\left ( \left \lceil \frac{k}{r} \right \rceil - 1 \right ) (\delta - 1).
$$
Then $M$ has a collection of cyclic flats $F_1,\ldots,F_m$ such that
\begin{enumerate}[(i)]
\item $\{F_i\}_{i \in [m]}$ has a non trivial union,
\item $|F_i| \leq r + \delta - 1$ for $i = 1,\ldots,m$,
\item $\eta(F_i) = \delta - 1 $,
\item $\{X \in \cflat(M) : X \subseteq F_i\} = \{\emptyset, F_i\}$ for $i = 1,\ldots,m$,
\item $\bigcup_{i \in [m]} F_i = E$,
\item statements c)--f) in Theorem \ref{th:structure-optimal-matroid} holds for every collection of flats $\{F_i\}_{i \in I}$ with $I \subseteq [m]$ and $|I| \leq \lkr$,
\item $k \leq  | \bigcup_{i \in I \setminus \{j\}} F_i | + |F_j| - \lkr(\delta - 1) -  | (\bigcup_{i \in I \setminus \{j\}} F_i) \cap F_j |$ for $I \subseteq [m]$, $|I| = \lkr$ and $j \in I$.
\end{enumerate} 
\end{corollary}

\begin{proof}
The statements (i)--(v) follows directly from Theorem \ref{th:structure-optimal-matroid}~(i)--(ii) and Theorem \ref{theorem:(n,k,d,r,delta)-Z}~(iv). Statement (vi) is a consequence of (i) and Theorem \ref{th:structure-optimal-matroid}~(iii), since (i) implies that $\{F_i\}_{i \in I}$ has a non trivial union. For statement (vii) we first observe by (iv), (vi) and Proposition \ref{pro:basic-Z}~(iii) that 
$$
 (\bigvee_{i \in I \setminus \{j\}} F_i) \wedge F_j =  ( \bigcup_{i \in I \setminus \{j\}} F_i ) \wedge F_j = \emptyset.
$$
Hence, by (vi) and axiom (Z3) in Theorem \ref{th:Z-axiom},
$$
\begin{array}{rcl}
k & = & \rho(\bigvee_{i \in I}F_i) \\
  &\leq& \rho(\bigvee_{i \in I \setminus \{j\}} F_i) + \rho(F_j) - \rho(\emptyset) - |(\bigvee_{i \in I \setminus \{j\}} F_i) \cap F_j|\\
  & = & | \bigcup_{i \in I \setminus \{j\}} F_i | - (\lkr - 1)(\delta - 1)+ |F_j| - (\delta - 1) -  | (\bigcup_{i \in I \setminus \{j\}} F_i) \cap F_j |.
\end{array}
$$
\end{proof}

%Note that even if the statements (i)--(vii) are satisfied in the corollary above by a subset of cyclic flats $F_1,\ldots,F_m$ of a matroid $M$, the lattice of cyclic flat $\cflat$ can still contain cyclic flats such that the parameter $d$ of the matroid does not meet the generalized Singleton bound.

We remark that structure theorems similar in spirit to the above have been given for linear $(n,k,d,r,\delta)$-LRCs  in \cite{Gopalan} and \cite{kamath14}. Namely, Theorem 2.2 in \cite{kamath14}  covers the case when $r | k$, showing that local repair sets correspond to linear $[r+\delta - 1,r,\delta]$-MDS codes and are mutually disjoint.
%Moreover, the theorem also states that for the column spaces $V_1, \ldots,V_t$ of the generator matrices to any $t$ MDS-codes that corresponds to $t$ linear repair sets, 
%$$
%\mathrm{dim}(V_t \cap (\sum_{i = 1}^{t-1} V_i)) = 0
%$$
%when $t = \kr$. 
Theorem 7 in \cite{Gopalan} proves the same in the special case $\delta = 2$.

Corollary \ref{corollary:structure-optimal-matroid}~(iv) means that the local matroid $M|{F_i}$ is uniform of rank $|F_i| - (\delta - 1)$, for $i=1,\ldots,m$. When the matroid comes from a linear code, the code in question is thus an $[|F_i|,|F_i|-(\delta-1),\delta]$-MDS code. By (vi) and (vii) in Corollary \ref{corollary:structure-optimal-matroid}, we obtain conditions on how large the intersections of union of subsets of the cyclic flats $\{F_i\}_{i \in [m]}$ can be. These results imply the corresponding results on linear LRCs. 
%Observing that flats for matroids corresponds to column spaces of a generator matrix of a linear code, we get the similar results for $\{F_i\}_{i \in {m}}$ when $r|k$ as the above described results given in \cite{kamath14}. In summery, our structure results are also valid when $r \nmid k$. In these cases we may afford some intersection between the repair sets. Further, the structure results are valid for matroids, for which the class of matroids associated to linear codes is a very small subclass.

\section{Constructions and classes of $\apart$-Matroids}

 The generalized Singleton bound theorem for matroids gives an upper bound for the value of $d$ in terms of the parameters $(n,k,r,\delta)$ for a matroid. In subsection \ref{subsec:4_constructions_matroid} we will give some constructions on $\atms$. These constructions will then be used in Subsection \ref{subsec:max_d_matroid}, where we will investigate, given different classes of the parameters $(n,k,r,\delta)$, whether or not perfect $(n,k,r,\delta)$-matroids exist.

\subsection{Combinatorial constructions of $\atms$} \label{subsec:4_constructions_matroid}

In this section we will give four increasingly specialized constructions of $\atms$. The constructions are purely combinatorial, and proceed by assigning the atomic cyclic flats, together with the rank function on the lattice of cyclic flats. In Section \ref{sec:applications}, we prove that the matroids we have constructed can be represented by linear codes.
%sentedfirst and second constructions are mainly obtained from sets, together with their intersections and ranks. The third and fourth construction are obtained from graphs. For $i = 1,2,3$, the matroids constructed in the $(i+1)$:th construction is a special class of the matroids constructed in the $i$:th construction. Of especial importance for us are construction two to four, since these will be used in Subsection \ref{subsec:max_d_matroid} and they can also be associated to linear codes, which will be proved and used in Section \ref{sec:applications}

\subsubsection{General construction of $\atms$}\label{sec:general}

%The basic building blocks from matroid theory to the two set constructions are cyclic flats. These cyclic flats are constructed from some chosen sets with chosen ranks.    

%\emph{The set construction 1:} 
%Let $F_1, \ldots, F_m$ be a collection of subsets of a finite set $E$, $F_I = \bigcup_{i \in I} F_i$ for $I \subseteq [m]$ and $\eta(F_i) = |F_i| - \rho(F_i) \hbox{ for } i \in [m]$.
%Further, let $k$ be a nonnegative integer and $\rho$  a function $\rho: \{F_i\}_{i \in [m]} \rightarrow \mathbb{Z}$ such that 
%\begin{equation} \label{eq:set_construction_1}
%\begin{array}{rl}
%(i) & \hbox{$\{F_i\}_{i\in [m]}$  has a nontrivial union 
%  with $F_{[m]}=E$,}\\
%(ii) & \hbox{$0 < \rho(F_i) < |F_i|$ for every $i \in [m]$,}\\
%(iii) & \hbox{There exists $I \subseteq [m]$ such that  $|F_I| - \sum_{i \in I} \eta(F_i) \geq k$,}\\
%(iv) & \hbox{If $F_I \in \mathcal{Z}_{<k}$ and $j \in [m] \setminus I$, then $|F_I \cap F_j| < \rho(F_j)$,}\\
%(v) & \hbox{If $F_I, F_J \in \mathcal{Z}_{<k}$ and $F_{I \cup J} \notin \mathcal{Z}_{<k}$, then $|F_{I \cup J}| - \sum_{t \in I \cup J} \eta(F_t) \geq k$,}
%\end{array}
%\end{equation}
%where 
%$$
%\mathcal{Z}_{<k} = \{F_J : J \subseteq [m] \hbox{ and } |F_I| - \sum_{i \in I} \eta(F_i) <k \hbox{ for all } I \subseteq J \}.
%$$ 

%
%A subclass of the matroids that we obtain by Theorem \ref{thm:construction} is given in the following construction.\\
 
%\emph{The set construction 2:} 
Let $F_1, \ldots, F_m$ be a collection of subsets of a finite set $E$ and define $F_I = \bigcup_{i \in I} F_i$ for $I \subseteq [m]$.
Further, let $k$ be a nonnegative integer and $\rho$  a function $\rho: \{F_i\}_{i \in [m]} \rightarrow \mathbb{Z}$ satisfying

\begin{enumerate}[(i)]
\item $0 < \rho(F_i) < |F_i|$ for $i \in [m]$,
\item $F_{[m]} = E$,
\item $k \leq |F_{[m]}| + \sum_{i \in [m]} (\rho(F_i)-|F_i|)$,
\item $I \subseteq [m], j \in [m] \setminus I \Rightarrow |F_I \cap F_j| < \rho(F_j)$.
\end{enumerate}

Define $$\mathcal{Z}_{<k}=\{F_J : |F_J| + \sum_{i \in J} (\rho(F_i)-|F_i|) <k\}$$ and $\mathcal{Z} = \mathcal{Z}_{<k} \cup \{E\}$.

Now, we extend the function $\rho$ to a function $\mathcal{Z}\to \mathbb{Z}$, by 
\begin{equation} \label{eq:Z_rho_construction}
\left \{
\begin{array}{lcl}
\rho(F_J) & = &|F_J| + \sum_{i \in J} (\rho(F_i)-|F_i|) \hbox{ for } F_J \in \mathcal{Z}_{<k},\\
\rho(E) &= & k.
\end{array} 
\right .
\end{equation}
Note that the extension of $\rho$ given in \eqref{eq:Z_rho_construction} is well defined, as by~(iii), $E$ is not in $\mathcal{Z}_{<k}$. Also note that $F_\emptyset = \emptyset$ and $\rho(F_\emptyset) = 0$.
Finally, we define $\mathcal{I} = \{X \subseteq E : |F_I \cap X| \leq \rho(F_I) \hbox{ for all } I \subseteq [m] \}$.

\begin{thm} \label{thm:construction}
Let $F_1,\ldots,F_m$ be a collection of subsets of a finite set $E$, $k$ a nonnegative integer and $\rho: \{F_i\}_{i \in [m]} \rightarrow \mathbb{Z}$ a function satisfying (i)--(iv). Then $\mathcal{Z}$ and $\rho: \mathcal{Z} \rightarrow \mathbb{Z}$, defined in \eqref{eq:Z_rho_construction}, define an $(n,k,d,r,\delta)$-matroid $M(F_1,\ldots,F_m;k;\rho)$ on $E$ for which $\mathcal{Z}$ is the collection of cyclic flats, $\rho$ is the rank function restricted to the cyclic flats, $\mathcal{I}$ is the set of independent sets, and 
\begin{enumerate}[(i)]
\item $n = |E|$,
\item $k = \rho(E)$,
\item $d = n-k+1-\max \{ \sum_{i \in I} \eta(F_i) : F_I \in \mathcal{Z}_{<k}\}$,
\item $\delta = 1+ \min_{i \in [m]} \{ |F_i|-\rho(F_i)\}$,
\item $r = \max_{i \in [m]} \{\rho(F_i)\}$.
\end{enumerate}
%For each $i \in [m]$, any subset $S \subseteq F_i$ with $|S| = \rho(F_i) + \delta -1$ is a locality set of the matroid.
\end{thm}

\begin{proof}
The proof is given in the Appendix.
\end{proof}

\begin{exam} \label{exam:graph_construction_1} 
Let $F_1, F_2, F_3$ be disjoint sets of cardinality $4$, with $\rho(F_1) = \rho(F_2) = 3$ and $\rho(F_3) = 2$. Moreover, let $(k,r,\delta) = (7,3,3)$. By Theorem \ref{thm:construction}, this corresponds to a matroid of size $14$ and minimum distance $4$.
\end{exam}

\subsubsection{Specialized construction of $\atms$}\label{sec:special}

To construct $\atms$ with large~$d$ in Section \ref{subsec:max_d_matroid}, we will use a special case of the construction in~\ref{thm:construction}. We represent the atomic cyclic flats $F_i$ by nodes in a graph, with labelled edges representing the intersections between the flats. The construction of a lattice of cyclic flats from a weighted graph can be made much more general by assigning weights to the nodes, representing the size and rank of the corresponding flats. However, in this section we specialize all parameters to obtain matroids that achieve the Singleton bound.    

Let $G$ be a graph with vertices $[m]$ and edges $W$, and let $\gamma:W\to\mathbb{Z}_{\geq 1}$ be a positive integer-valued function on the edge set. Moreover, let $(k,r,\delta)$ be three integers with $0 < r < k$ and $\delta \geq 2$,  such that 

\begin{enumerate}[(i)]\label{eq:graph_1}
\item $G$ has no triangles,
%\item $0 \leq \alpha(i) \leq r-1$ for $i \in [m]$,
\item $k \leq rm - \sum_{w \in W} \gamma(w)$,
\item $r > \sum_j \gamma(\{i,j\})$ for every $i \in [m]$.
\end{enumerate}

From the graph $G$ we construct the sets $F_1,\ldots,F_m$ and the rank function $\rho$ by first assigning the following: 
\begin{enumerate}[(i)]
\setcounter{enumi}{3}
\item $\rho(F_i) = r$ for $i \in [m]$,
\item $|F_i| = r + \delta - 1$ for $i \in [m]$, 
\item $|F_i \cap F_j| = \gamma(\{i,j\})$ for $\{i,j\} \in W$.
\end{enumerate} 

Note that (v)--(vii) uniquely defines the sets $F_1,\ldots,F_m$ and their ranks, up to isomorphism. This can be seen by induction over $m$, observing that (iv) guarantees that the intersections $F_i\cap F_j$ can be chosen to be disjoint for different $j$. This is required, as there is no 3-cycle in the graph $G$, so
$$
|F_h \cap F_i \cap F_j | = 0 \hbox{ for all three distinct elements } h,i,j \in [m].
$$
Also note that, while $n$ is not a parameter of the graph construction, it is a function of the parameters, as we have $$n=|\cup_i F_i| = m(r+\delta-1)-\sum_{w\in W} \gamma(w).$$

\begin{thm} \label{theorem:graph_1}
Let $F_1,\ldots,F_m$ and $\rho:\{F_i\}\to\mathbb{Z}$ be constructed from a weighted graph $(G,\gamma)$ with parameters $(k,r,\delta)$ according to (i)--(vi). Then $(\{F_i\},\rho)$ satisfies (i)--(iv) in~\ref{sec:general}.
In particular, $\{F_i\}$ are the atomic cyclic flats of an $(n,k,d,r,\delta)$-matroid with
\begin{enumerate}[(i)]
\item $n = (r + \delta - 1)m - \sum_{w \in W} \gamma(w)$,
\item $d = n - k + 1 - (\delta-1)\max \{|I| : r |I| - \sum_{w\in W\cap I\times I} \gamma(w) < k\}$.
\end{enumerate}
\end{thm}

\begin{proof}
The proof is given in the Appendix.
\end{proof}

%Example \ref{exam:graph_construction_1} and \ref{exam:graph_construction_1b} below gives two examples of $\atms$ that can be obtained from the graph construction given in Theorem \ref{theorem:graph_1}.
%

%all other values of $\alpha$ and $\beta$ constantly zero.

%By Theorem \ref{thm:construction}, this corresponds to a $(14,7,4,3,3)$-matroid $M(F_1,F_2,F_3;k;\rho)$, where $F_1$, $F_2$ and $F_3$ are pairwise disjoint sets with 
%$|F_1|=|F_2| = 5$, $|F_3|=4$,
%$\rho(F_1) = \rho(F_2) = 3$ and $\rho(F_3) = 2$.
%\end{exam}

%
%\begin{exam} \label{exam:graph_construction_1b} 
%Let $G(\alpha,\beta,\gamma;10,5,3)$ denote the graph with vertex set $m=[4]$ and edge set $W = \{\{1,2\}, \{2,3\}, \{3,4\}, \{4,1\}\}$, $(k,r,\delta) = (10,5,3)$, $(\alpha(1), \alpha(2), \alpha(3), \alpha(4)) = (1,0,1,0)$, $(\beta(1), \beta(2), \beta(3), \beta(4)) = (0,1,0,2)$ and $(\gamma(\{1,2\}), \gamma(\{2,3\}), \gamma(\{3,4\}), \gamma(\{4,1\})) = (1,2,1,1)$. By the statement (i)--(xi) above, this gives 
%$$
%\begin{array}{l}
%(\rho(F_1),\rho(F_2),\rho(F_3),\rho(F_4)) = (4,5,4,5),\\
%(|F_1|,|F_2|,|F_3|,|F_4|) = (6,8,6,9),\\
%F_1 = \{1,\ldots,6\} \com F_2 = \{1,7, \ldots,13\} \com F_3 = \{7,8,14,\ldots,17\} \com F_4 = \{2,14,18,\ldots,24\}.
%\end{array}
%$$

%By Theorem \ref{theorem:graph_1} 
%$$
%n = (5+3-1)4 - 2 + 3 - 5 = 24
%$$
%and as $V_{< k } = \{X \subseteq [4] : |X| \leq 2 \hand X \neq \{2,4\}\}$ it follows that 
%$$
%d = 24 - 10 + 1 - (2\cdot 2 + 2) = 9. 
%$$
%\end{exam}
%
  
Now, in addition, we assume that $G$ has girth at least $\max \{4, \kr+1\}$, and that the weight function $\gamma$ does not take too large values. Then we get the following theorem, on the existence of perfect $(n,k,r,\delta)$-matroids.
%\begin{equation} \label{eq:graph_2}
%\begin{array}{rl}
%(i) & \hbox{$G$ is a graph with no $l$-cycles, for $l \leq \max \{3, \lceil \frac{k}{r} \rceil \}$,}\\
%(ii) & \hbox{$\gamma(w) \geq 1$ for $w \in W$,}\\
%(iii) & \hbox{$\sum_{w \in W} \gamma(w) = b$,}\\
%(iv) & \hbox{$k \leq rm - b$,}\\
%(v) & \hbox{For $I \subseteq [m]$ with $|I| = \lceil \frac{k}{r} \rceil$, $\sum_{w\in W\cap I\times I} \gamma(w) \leq a$,}\\
%(vi) & \hbox{$r > \sum_{j} \gamma(\{i,j\})$ for $i \in [m]$.}
%\end{array}
%\end{equation} 

%From the graph $G$ we construct the sets $F_1,\ldots,F_m$ and the rank function by first assigning $\alpha(i) = \beta(i) = 0$ for $i = 1,\ldots,m$ in \eqref{eq:M_from_graph_1} (i.e. $\rho(F_i) = r$ and $|F_i| = r+\delta-1$) and then use \eqref{eq:construct_F_graph_1} to recursively construct the sets $F_1,\ldots,F_m$.

\begin{corollary} \label{theorem:graph_2}
Let $(G,\gamma)$ be a weighted graph, and let $(k,r,\delta)$ be integers such that (i)--(iii) is satisfied. Let $b=\sum_{w \in W} \gamma(w)$, and $a=\kr r - k$. Assume moreover that $G$ has no $l$-cycles, for $l \leq \lceil \frac{k}{r} \rceil$, and that $\sum_{w\in W\cap I\times I} \gamma(w) \leq a$ for every $I \subseteq [m]$ with $|I| = \lceil \frac{k}{r} \rceil$.

Then there exists a $(n,k,d,r,\delta)$-matroid with 
\begin{enumerate}[(i)]
\item $n = (r + \delta - 1)m - b$,
\item $d = n - k +1 - (\lceil \frac{k}{r} \rceil - 1)(\delta - 1)$.
\end{enumerate}
\end{corollary} 

\begin{proof}
We need to prove that $$\kr - 1\leq \max \left\{|I| : r |I| - \sum_{w\in W\cap I\times I} \gamma(w) < k=\kr r - a\right\}.$$ If $|I| = \kr -1$, then $$r |I| - \sum_{w\in W\cap I\times I} \gamma(w)\leq r(\kr-1)<k.$$ If, on the other hand, $|I| = \kr$, then $$r |I| - \sum_{w\in W\cap I\times I} \gamma(w)= r\kr - \sum_{w\in W\cap I\times I}\gamma(w)\geq r\kr-a =k,$$ by assumption. Thus, the corollary follows from Theorem~\ref{theorem:graph_1}.
\end{proof}

\begin{corollary} \label{corollary:graph_2}
Let $(G,\gamma)$ be a weighted graph, and let $(k,r,\delta)$ be integers such that (i)--(iii) is satisfied. Let $b=\sum_{w \in W} \gamma(w)$, and $a=\lceil \frac{k}{r} \rceil r - k$. Assume moreover that $G$ has no $l$-cycles, for $l \leq \lceil \frac{k}{r} \rceil$, and that $1\leq \gamma(w) \leq \left \lfloor \frac{a}{\kr - 1} \right \rfloor$ for every $w\in W$. 
Then there is an $(n,k,d,r,\delta)$-matroid with 
\begin{enumerate}[(i)]
\item $n = (r + \delta - 1)m - b$,
\item $d = n - k +1 - (\lceil \frac{k}{r} \rceil - 1)(\delta - 1)$.
\end{enumerate}

\end{corollary}

\begin{proof}
Since $G$ has no $l$-cycles for $l \leq \kr$, we have for every $I \subseteq [m]$ with $|I| = \kr$ that $|W\cap I\times I|\leq \kr - 1$. Since $\gamma(w) \leq \left \lfloor \frac{a}{\kr - 1} \right \rfloor$, we then get
$$
\sum_{w\in W\cap I\times I} \gamma(w) \leq \left \lfloor \frac{a}{\kr - 1} \right \rfloor \left (\lkr - 1 \right ) \leq a, 
$$
so Theorem~\ref{theorem:graph_2} applies. \end{proof}

We remark that in order to find as small $n$ as possible for a chosen $(k,r,\delta,a,b)$ in Corollary \ref{corollary:graph_2}, we want to find a good graph with as few nodes as possible. To find such a graph, preferable properties for the graph are: many small cycles of length $\max \{4, \kr + 1\}$, large values of $\gamma$ on every edge, \emph{i.e.} $\gamma(w) = \left \lfloor \frac{a}{\kr - 1} \right \rfloor$ for $w \in W$, and that the sum of $\gamma$-values incident to each node is large, \emph{i.e.} $\sum_{j} \gamma(\{ i,j\}) = r-1$ for all nodes $i \in [m]$). 

\begin{exam} \label{exam:graph_construction_2a}
Let $G$ denote the graph below on the vertex set $[6]$, where the values of $\gamma$ are written above the edges in the graph, and $(k,r,\delta)=(14,4,2)$. We get $b=\sum\gamma(w)=3$ and $a=r\kr -k=2$
\begin{figure}[htb]
    \centering
    \includegraphics[width=0.36\textwidth]{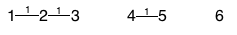}
    \caption{The graph $G(\gamma;14,4,2,2,3)$}
    \label{fig:smallgraph}
\end{figure}

 By Corollary \ref{corollary:graph_2}, %with $\rho(F_i) = r = 4$ and $|F_i|= r+\delta - 1 = 5$ for $i=1,\ldots,6$ in \eqref{eq:construct_F_graph_1}, 
this graph corresponds to a $(27,14,11,4,2)$-matroid on the ground set $[27]$, with six atomic cyclic flats $F_1,\ldots F_6$, where
$$
\begin{array}{l}
F_1 = \{1,\ldots,5\} \hbox{, } F_2 = \{1,6,\ldots,9\} \hbox{, } F_3 = \{6,10,\ldots,13\} \hbox{, }\\
F_4 = \{14,\ldots,18\} \hbox{, } F_4 = \{14,19,\ldots,22\} \hbox{ and } F_5 = \{23,\ldots,27\}.
\end{array}
$$  
\end{exam}

\begin{exam} \label{exam:graph_construction_2b}
Let $G = G(\gamma;k,r,\delta,a,b)$ denote the graph below on the vertex set $[11]$. The $\gamma$-values for the edges are written in the graph and $(k,r,\delta, a,b) = (19,9,5,8,21)$.

\begin{figure}[!htb]
    \centering
    \includegraphics[width=0.35\textwidth]{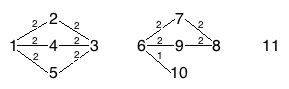}
    \caption{The graph $G(\gamma;19,9,5,8,21)$}
    \label{fig:largegraph}
\end{figure}

By Corollary \ref{corollary:graph_2}, this graph corresponds to a $(122,19,96,9,5)$-matroid, whose lattice of cyclic flats has $11$ atoms. \end{exam}

\subsection{The maximal $d$ for $(n,k,r,\delta)$-matroids} \label{subsec:max_d_matroid}

We know by Theorem \ref{th:bound}, that the inequality $$d \leq n-k- (\left \lceil \frac{k}{r} \right \rceil - 1)(\delta - 1)$$ holds for any $(n,k,d,r,\delta)$-matroid. It is then very natural to ask what is the maximal value of $d$, for which there exists an $(n,k,r,\delta)$-matroid, for given $(n,k,r,\delta)$ with $0 <r\leq k \leq n- (\delta-1) \lceil \frac{k}{r} \rceil$ and $\delta \geq 2$. We will denote this maximal value $\dm=\dm(n,k,r,\delta)$.  The case $r=k$ is degenerate, and we will consider this first. The case when $r<k$ will be further divided into four subcases in Theorem \ref{thm:max_d}. Theorem \ref{thm:max_d} will later translate into results for linear LRCs in Theorem \ref{thm:non_existence_LRC} and Theorem \ref{thm:max_d_LRC}.

\subsubsection{The maximal value of $d$ when $r=k$}  
\label{sec:max_d_r=k} 

A well known class of matroids is the class of \emph{uniform matroids} \cite{oxley92},  defined as $U_n^k = (\rho,E)$, where
\begin{equation} \label{eq:def_uniform_matroid}
|E| = n \hand \rho(X) = \min\{|X| \com k\}.
\end{equation} 
This implies that the cyclic sets of $U_n^k$ is
$$
\mathcal{U}(U_n^k) = \{\emptyset\} \cup \{X \subseteq E : \lvert X \rvert \geq k + 1\},
$$ 
and that the cyclic flats are
$$
\mathcal{Z} = \{0_{\mathcal{Z}}, 1_{\mathcal{Z}}\}, \hbox{ with } 0_{\mathcal{Z}} = \emptyset \hbox{, } 1_{\mathcal{Z}} = E \hbox{, } \rho(0_{\mathcal{Z}}) = 0 \hbox{ and } \rho(1_{\mathcal{Z}}) = k.
$$

If $k=r$, the generalized Singleton bound given in Theorem \ref{th:bound} reduces to the classical Singleton bound, $d = n-k+1$. Then using Theorem \ref{theorem:(n,k,d,r,delta)-Z} (iii), we get that $\mathcal{Z} = \{\emptyset, E\}$, so $M$ is the uniform matroid $U_n^k$. For $(r,\delta)$-locality, let $S_x = U_n^k$ for each $x \in E$ and $\delta = d = n-k+1$. Then $\lvert S_x \rvert = r + (\delta-1)$ and $d(S_x) = \delta$. Consequently, $U_n^k$ is a matroid with parameters $(n,k,d,r,\delta) = (n,k,n-k+1,r,n-k+1)$.

\subsubsection{The maximal value of $d$ when $r<k$} \label{sec:r<k}

As the first result of this section, we prove that $$n-k- \left \lceil \frac{k}{r} \right \rceil(\delta - 1)\leq \dm(n,k,r,\delta) \leq n-k- (\left \lceil \frac{k}{r} \right \rceil - 1)(\delta - 1),$$ where the second inequality is Theorem~\ref{th:bound} revisited.  We will then use the graph constructions given in Theorem \ref{theorem:graph_1} and Theorem \ref{theorem:graph_2}, in order to construct matroids with larger $d$. In the cases when $\dm < n-k- (\left \lceil \frac{k}{r} \right \rceil - 1)(\delta - 1)$, we will use Theorem \ref{th:structure-optimal-matroid} to prove this.

\begin{thm}\label{thm:almosttight}
For any $(n,k,r,\delta)$ satisfying $1\leq r< k\leq n-\kr(\delta-1)$ and $2\leq \delta< n$, there exists a $(n,k,d,r,\delta)$-matroid, where $$d=n-k- \left \lceil \frac{k}{r} \right \rceil(\delta - 1).$$\end{thm}
\begin{proof}
Let $m=\left\lceil\frac{n}{r+\delta-1}\right\rceil$, and let $F_1,\ldots F_{m-1}$ be disjoint sets with rank $r$ and size $r+\delta-1$. Let $F_m$ be disjoint from all of $F_1,\ldots F_{m-1}$, with size $|F_m|=n-(m-1)(r+\delta-1)$ and rank $\rho(F_m)=|F_m|-\delta+1$. Finally, let $M$ be defined by $\mathcal{Z}(M)=\{F_I\}$, where $F_I=\cup_{i\in I}F_i$, and $$\rho(F_I)=\min\{\sum_{i\in I}\rho(F_i), k\}.$$ It is readily seen that $M$ has minimum distance $$d= n-k+1-(\delta-1)\max\{|I|: \rho(F_I)<k\}\geq n-k+1- \left \lceil \frac{k}{r} \right \rceil(\delta - 1).$$
\end{proof}

In particular, when $\delta=2$, this means that the optimal minimum distance is one of $n-k+1-\kr$ and $n-k-\kr$. The remainder of this section aims at deciding which of these two possibilities is the case for fixed 

 Before stating the technical theorem on $d_{\mathrm{max}}$, we need the following qualitative result.

\begin{proposition} \label{proposition:bound_frac_for_n}
Let $M$ be an $(n,k,d,r,\delta)$-matroid and let $a = \left \lceil \frac{k}{r} \right \rceil r - k$ and $b = \left \lceil \frac{n}{r+ \delta - 1} \right \rceil (r+ \delta - 1) - n$. Then the following hold,
$$
\left \lceil \frac{n}{r+ \delta - 1} \right \rceil \geq
      \left \{
			\begin{array}{lcl}
			\left \lceil \frac{k}{r} \right \rceil & \hbox{if} & b\leq a,\\
			\left \lceil \frac{k}{r} \right \rceil + 1 & \hbox{if} & b > a,
			\end{array}
			\right .
$$
\end{proposition}

\begin{proof}
Let $\left \lceil \frac{n}{r+ \delta - 1} \right \rceil = \left \lceil \frac{k}{r} \right \rceil + t$. Note that $n-k \geq \left \lceil \frac{k}{r} \right \rceil (\delta - 1)$ by Proposition \ref{proposition:bound_delta_k_rate}. Hence,
$$
\begin{array}{lcl}
\left \lceil \frac{k}{r} \right \rceil (\delta - 1) & \leq & n-k\\
 & = & (\left \lceil \frac{k}{r} \right \rceil + t) (r + \delta - 1) - b - (\left \lceil \frac{k}{r} \right \rceil r - a)\\
      & = & \left \lceil \frac{k}{r} \right \rceil (\delta - 1) + t(r + \delta - 1) -(b - a).
      \end{array}
$$
This implies that $t \geq 0$ if $b \leq a$ and $t \geq 1$ if $b > a$.
\end{proof}

\begin{thm} \label{thm:max_d}
Let $(n,k,r,\delta)$ be integers such that $0 < r < k \leq n- \left \lceil \frac{k}{r} \right \rceil(\delta-1)$, $k = \left \lceil \frac{k}{r} \right \rceil r - a$ and $n = \left \lceil \frac{n}{r+ \delta - 1} \right \rceil (r+ \delta - 1) - b$. Let $\dm = \dm (n,k,r,\delta)$ be the largest $d$ such that there exists an $(n,k,d,r,\delta)$-matroid. Then the following hold.
\begin{enumerate}[(i)]
\item If  $a \geq b$, then $\dm = n-k+1-\left (\left \lceil \frac{k}{r} \right \rceil - 1 \right) (\delta - 1)$;

\item If $b > a$ and $b\geq r$, then
$ \dm \geq n-k+1- \left \lceil \frac{k}{r} \right \rceil (\delta - 1) + (b - r).$

\item If $b > a$ and $a < \left \lceil \frac{k}{r} \right \rceil - 1$, then \\ 
$\dm = n-k+1-\left (\left \lceil \frac{k}{r} \right \rceil - 1 \right) (\delta - 1)$ if and only if 
 $\left \lfloor \left \lceil \frac{k}{r} \right \rceil / 2\right \rfloor \leq a$ and $$\left \lceil \frac{n}{r+\delta-1} \right \rceil \geq \left \lceil \frac{k}{r} \right \rceil -1 + \left(b-a\right)\left(1+\frac{1}{t}\right),$$ where $t= \left\lfloor a/\left(\left \lceil \frac{k}{r} \right \rceil -1-a\right)\right\rfloor$;

\item If $b > a \geq \left \lceil \frac{k}{r} \right \rceil - 1$, $\left \lceil \frac{k}{r} \right \rceil \geq 3$ and
$$\left \lceil \frac{n}{r+\delta-1} \right \rceil \geq \left \lfloor \frac{b}{stu} \right \rfloor \left ( t (u-1) + 2 \right) + y,$$
 where $s = \left \lfloor \frac{a}{\lceil \frac{k}{r} \rceil - 1} \right \rfloor$, $t = \left \lfloor \frac{r-1}{s} \right \rfloor$, $u = \left \lceil \frac{\kr + 1}{2}\right \rceil$,  $x = \left \lceil \frac{b - \lfloor \frac{b}{stu} \rfloor stu}{s} \right \rceil$, and  
$$	  y = 
\left \{
\begin{array}{lcl}
0 &\hbox{if}& stu \mid b,\\
x - \left \lfloor \frac{x}{u} \right \rfloor + 1 +\min\{\left \lfloor \frac{x}{u} \right \rfloor, 1\} &\hbox{if}& stu \nmid b,
\end{array}
\right .$$
then $\dm = n-k+1-\left (\left \lceil \frac{k}{r} \right \rceil - 1 \right) (\delta - 1)$;\\

\item If $b > a \geq \left \lceil \frac{k}{r} \right \rceil - 1$, $\left \lceil \frac{k}{r} \right \rceil = 2$, and $$\left \lceil \frac{n}{r+\delta- 1} \right \rceil \geq 
\left \{
\begin{array}{lcl}
\lceil \frac{b}{a}\rceil + 1 &\hbox{if}& 2a \leq r-1,\\
\left \lceil \frac{b}{\lfloor \frac{r-1}{2} \rfloor} \right \rceil + 1 &\hbox{if}& 2a > r-1,\\
\end{array}
\right . $$
then $\dm = n-k+1-\left (\left \lceil \frac{k}{r} \right \rceil - 1 \right) (\delta - 1)$.
\end{enumerate}
\end{thm}

\begin{proof}
The proof is given in the Appendix.
\end{proof}

In the proof of Theorem \ref{thm:max_d}(iv), we will notice that a simpler bound, but in general not as good, is $\left \lceil \frac{n}{r+\delta-1} \right \rceil \geq \left \lceil \frac{b}{stu} \right \rceil \left ( t (u-1) + 2 \right )$.

\begin{exam} \label{exam:d_max}
Examples of constructions of matroids in Theorem \ref{thm:max_d}(i), (iii) and (iv) given by the proofs of the theorem are given in Example \ref{exam:graph_construction_1}, \ref{exam:graph_construction_2a} and \ref{exam:graph_construction_2b} respectively.
\end{exam}

\section{Applications of $(n,k,d,r,\delta)$-Matroids to $(n,k,d,r,\delta)$-LRCs} \label{sec:applications}

In this section we will use the previous results on $(n,k,d,r,\delta)$-matroids to get new results on linear and almost affine $(n,k,d,r,\delta)$-LRCs. All the proofs of the non-existence of matroids immediately give corresponding bounds for codes. To verify the other direction, obtaining codes with prescribed parameter values from matroids with the same parameters, we will show that the class of matroids given in Theorem \ref{thm:construction} is a subclass of a class of matroids called gammoids.  Gammoids have the property that they are representable over any finite field of sufficiently large size. 

The main result in this section is Theorem \ref{thm:class_of_gammoids}. 
\begin{thm} \label{thm:class_of_gammoids}
Let $M(F_1,\ldots,F_m;k;\rho)$ be an $(n,k,d,r,\delta)$-matroid that we get in Theorem~\ref{thm:construction}. Then for every large enough finite field there is a linear LRC over the field with parameters $(n,k,d,r,\delta)$.  
\end{thm}

\subsection{Transversal matroids and gammoids}

We start by giving a short introduction to gammoids. For more information on this fascinating class of matroids we refer the reader to \cite{oxley92, schrijver03}.

A gammoid is associated to a directed graph $G$ as follows.

\begin{definition} Let $G=(V,D)$ be a directed graph, with $S \subseteq V$ and 
$T \subseteq V$.The  \emph{gammoid} $M(G)$ is a matroid $M(G)$ on $S$ where the independent sets of $M(G)$ equals
$$
\mathcal{I}(M(G)) = \{X \subseteq S : \exists \hbox{ a set of $|X|$ vertex-disjoint paths from $X$ to $T$}\}.
$$

\end{definition} 
Our interest in gammoids in this paper stems from the following result.

\begin{thm}[\cite{lindstrom73}] \label{thm:gammoid_representable}
Every gammoid over a finite set $E$ is representable over every finite field of size greater than or equal to $2^{|E|}$.
\end{thm}

Many natural classes of gammoids, can be represented over fields of much smaller size than $2^n$. For example, a uniform matroid $U_n^k$ \eqref{eq:def_uniform_matroid} is a gammoid associated to a complete bipartite graph with $V=S\cup T$, $|S=n|$, $|T|=k$ and $D=S\times T$. However, uniform matroids are represented by linear $[n,k,d=n-k+1]$-MDS codes, which exist over $\mathbb{F}_q$ when $q \geq n$.

%A \emph{transversal matroid} is a gammoid $M(G')$ over a directed graph $G' = (V',D',S',T')$, where $S'$ and $T'$ are disjoint sets, $V' =  S' \cup T'$ and the set of directed edges $D'$ is a subset of $\{(\overrightarrow{s,t}) : s \in S' \hbox{ and } t \in T'\}$. By the Hall's Theorem \cite[Thm. 12.2.1]{oxley92} the independent sets of a transversal matroid $M(G')$ is
%\begin{equation} \label{eq:hall}
%\mathcal{I}(M(G')) = \{X \subseteq S' : |X'| \leq |A(X')| \hbox{ for each } X' \subseteq X\},
%\end{equation}
%where $A(X') = \{t \in T : \exists x' \in X' \hbox{ such that } (\overrightarrow{x',t}) \in D'\}$.

\subsection{Constructions of linear $(n,k,d,r,\delta)$-LRCs $C(F_1, \ldots,F_m;k;\rho)$}\label{sec:representability}
Theorem~\ref{thm:class_of_gammoids} follows immediately from Lemma~\ref{lemma:gammoid_graph} and Theorem~\ref{thm:gammoid_representable}. The key element is the construction of a directed graph whose associated gammoid is the matroid from Theorem~\ref{thm:construction}. This construction is detailed in Algorithm 1.

\begin{algorithm} [!hbt]
\caption{Input: $(F_1,\ldots,F_m;k;\rho)$. Output: $G= (V,D,S,T)$ }\label{euclid}
\label{alg:transversal}
\begin{algorithmic}[1]
\State $S = E$, $H = \emptyset$, $D = \emptyset$, $T=[k]$
\State Label $e\in S$ with $s(e)=\{i:e\in F_i\}$.
\State $h \textit{ is a function } H\to 2^{[m]}$

\ForAll {$e \in E$}
\If {$|s(e)| \geq 2$}
\State $ H \gets H \cup \{u_e\} $
\State $h(u_e) = s(e)$
\EndIf
\EndFor
\ForAll {$i \in [m]$}
\State $l_i = \rho(F_i) - |\{u \in H : i \in h(u)\}|$
\State $H \gets H \cup \{v_1^{i}, \ldots, v_{l_i}^{i}\}$
\State $ h(v_1^{i}) = \ldots = h(v_{l_i}^{i}) = \{i\} $
\EndFor
\ForAll {$(e,u) \in S \times H$}
\If {$s(e) \subseteq h(u)$}
\State $D \gets D \cup (\overrightarrow{e,u})$
\EndIf
\EndFor
%\ForAll {$(s,u) \in S \times H$}
\State $D \gets D \cup H\times T$
%\EndFor
\State $V = S \cup (H\cup T)$
\State Output $(V,D,S,T)$
\end{algorithmic}
\end{algorithm}

\begin{lemma} \label{lemma:gammoid_graph}
Let $F_1,\ldots,F_m$ be a collection of subsets of a finite set $E$ whose union is all of $E$, and write $F_I=\cup_{i\in I} F_i$. Let $\rho: \{F_i\}_{i \in [m]} \rightarrow \mathbb{Z}$ satisfy
\begin{enumerate}[(i)] 
\item $0\leq\rho(F_i)\leq |F_i|$,
\item $k\leq |F_{[m]}|+\sum_{i\in [m]}(\rho(F_i)-|F_i|)$,
\item $|F_I\cap F_j|< \rho(F_j)$ whenever $j\not\in I$. 
\end{enumerate}
 Then the gammoid $M(G)$, that we get from Algorithm \ref{alg:transversal} is equal to the matroid $M(F_1,\ldots,F_m;k;\rho)$ that we get in Theorem\ref{thm:construction}. 
\end{lemma}

\begin{proof}
The proof is given in the Appendix.
\end{proof}

%\begin{algorithm} [!htb]
%\caption{From $(F_1,\ldots,F_m;k;\rho)$ to $G(F_1,\ldots,F_m;k;\rho)= (V,D,S,T)$ }\label{euclid}
%\label{alg:gammoid}
%\begin{algorithmic}[1] 
%\State \textit{Get the graph }$ G'(F_1,\ldots,F_m;\rho)= (V',D',S', T')$ \textit{ from Algorithm \ref{alg:transversal}}  
%\State $H = T'$, $S = S'$, $D = D'$
%\State $T = \{t_1,\ldots, t_k \} $
%\State $V = V' \cup T$
%\ForAll {$(u,t) \in H \times T$}
%\State $D \gets D \cup (\overrightarrow{u,t})$
%\EndFor
%\State $G(F_1,\ldots,F_m;k;\rho) = (V,D,S,T)$
%\end{algorithmic}
%\end{algorithm}

%
%\begin{lemma} \label{lemma:gammoid_graph}
%Let $F_1,\ldots,F_m$ be a collection of subsets of a finite set $E$, $k$ a nonnegative integer and $\rho: \{F_i\}_{i \in [m]} \rightarrow \mathbb{Z}$ a function such that the statements (i)--(iv) \eqref{eq:set_construction_2} are satisfied. Then the gammoid $M(G(F_1,\ldots,F_m;k;\rho))$, that we get from Algorithm \ref{alg:gammoid}, is equal to the matroid $M(F_1,\ldots,F_m;k;\rho)$ that we get in Theorem \ref{thm:construction}. 
%\end{lemma}
%

%\begin{proof}
%The proof is given in the Appendix.
%\end{proof}

\begin{exam}
Consider the matroid $M_G$, associated to the storage system in Figure~\ref{fig:cloud} and the code in Example~\ref{exam:matroid}, and whose lattice of cyclic flats is written out in Example~\ref{exam:Z}. By Lemma~\ref{lemma:gammoid_graph}, this is the gammoid associated to the following graph, with $|T|=k=6$ and $|S|=n=12$. 
Note that in this particular setting, Line 15 in Algorithm~\ref{alg:transversal} is superfluous, could be replaced by assigning $H=T$, since $H$ already has only $6$ nodes. Indeed, the inclusion of the bipartite graph $(H,T)$ corresponds to truncating the gammoid at rank $k$.
\begin{figure}[htb]
    \centering
    \includegraphics[width=0.7\textwidth]{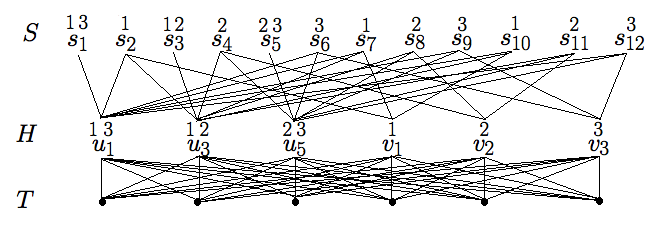}
    \caption{A downward directed graph supporting the matroid associated to a $(12,6,3,3,3)$-LRC.} 
    \label{fig:gammoid}
\end{figure}
\end{exam}

\subsection{Bounds on the parameters $(n,k,d,r,\delta)$ for LRCs}

In this section we will give results on the parameters $(n,k,d,r,\delta)$ for linear, and more generally almost affine LRCs. The results are direct consequences of the corresponding results for matroids, thanks to the representability results in Theorem \ref{thm:class_of_gammoids} and the matroid invariance of the parameters $(n,k,d,r,\delta)$, from Theorem \ref{thm:matroid_invariant}. We will therefore not give any further proofs in this section. Observe that this means that the same bounds are valid for matroids, almost affine codes, and linear codes. 

\begin{thm} \label{thm:singleton_almost_affine}
If $C$ is an almost affine LRCs with the parameters $(n,k,d,r,\delta)$, then
$$
d \leq n-k+1 -  \left ( \left \lceil \frac{k}{r} \right \rceil - 1 \right )(\delta-1).
$$
\end{thm}

\begin{proposition} \label{prop:delta_k_rate_LRC}
Let $C$ be an almost affine LRC with parameters $(n,k,d,r,\delta)$. Then 
\begin{enumerate}[(i)]
\item $\delta \leq d$,
\item $k \leq n - \left \lceil \frac{k}{r} \right \rceil (\delta - 1)$,
\item $\frac{k}{n} \leq \frac{r}{r + \delta - 1}$.
\end{enumerate}
\end{proposition}

\begin{thm} \label{thm:non_existence_LRC}
Let $C$ be an almost affine LRC with parameters $(n,k,r,d,\delta)$, and let $a = \left \lceil \frac{k}{r} \right \rceil r - k$ and $b = \left \lceil \frac{n}{r+ \delta - 1} \right \rceil (r+ \delta - 1) - n$. Then the following hold.
\begin{enumerate}[(i)]
\item If $b>a$ and $a < \left \lfloor \left \lceil \frac{k}{r} \right \rceil / 2\right \rfloor$, then $d < n-k+1-\left (\left \lceil \frac{k}{r} \right \rceil - 1 \right) (\delta - 1)$;\\
\item If $b>a$ and $\left \lfloor \left \lceil \frac{k}{r} \right \rceil / 2\right \rfloor \leq a \leq \kr - 1$, then $d < n-k+1-\left (\left \lceil \frac{k}{r} \right \rceil - 1 \right) (\delta - 1)$.
\end{enumerate}
\end{thm}

\begin{thm} \label{thm:max_d_LRC}
Let $(n,k,r,\delta)$ be integers such that $0 < r < k \leq n- \left \lceil \frac{k}{r} \right \rceil(\delta-1)$, $a = \left \lceil \frac{k}{r} \right \rceil r - k$ and $b = \left \lceil \frac{n}{r+ \delta - 1} \right \rceil (r+ \delta - 1) - n$. Let $\dm = \dm (n,k,r,\delta)$ be the largest $d$ such that there exists a linear LRC with parameters $(n,k,d,r,\delta)$. Then the following hold.
\begin{enumerate}[(i)]
\item If $a \geq b$, then $\dm = n-k+1-\left (\left \lceil \frac{k}{r} \right \rceil - 1 \right) (\delta - 1)$;
\item If $b > a$, then 
$$\dm \geq \left \{
		   \begin{array}{lcl}
			 n-k+1- \left \lceil \frac{k}{r} \right \rceil (\delta - 1) &\hbox{if}&  b \leq r-1,\\
			 n-k+1- \left \lceil \frac{k}{r} \right \rceil (\delta - 1) + (b - r) &\hbox{if}&  b \geq r;
			 \end{array}
		   \right .$$

\item If $b > a$, $\left \lfloor \left \lceil \frac{k}{r} \right \rceil / 2\right \rfloor \leq a < \left \lceil \frac{k}{r} \right \rceil - 1$ and $\left \lceil \frac{n}{r+\delta-1} \right \rceil \geq \left \lceil \frac{k}{r} \right \rceil -1 + \left(b-a\right)\left(1+\frac{1}{t}\right)$,
where $t= \left\lfloor a/\left(\left \lceil \frac{k}{r} \right \rceil -1-a\right)\right\rfloor$, then $$\dm = n-k+1-\left (\left \lceil \frac{k}{r} \right \rceil - 1 \right) (\delta - 1);$$

\item If $b > a \geq \left \lceil \frac{k}{r} \right \rceil - 1$, $\left \lceil \frac{k}{r} \right \rceil \geq 3$ and $\left \lceil \frac{n}{r+\delta-1} \right \rceil \geq \left \lfloor \frac{b}{stu} \right \rfloor \left ( t (u-1) + 2 \right) + y$, \\
where $s = \left \lfloor \frac{a}{\lceil \frac{k}{r} \rceil - 1} \right \rfloor$, $t = \left \lfloor \frac{r-1}{s} \right \rfloor$, $u = \left \lceil \frac{\kr + 1}{2}\right \rceil$, $x = \left \lceil \frac{b - \lfloor \frac{b}{stu} \rfloor stu}{s} \right \rceil$ and   
$$	  y = 
\left \{
\begin{array}{lcl}
0 &\hbox{if}& stu \mid b,\\
x - \left \lfloor \frac{x}{u} \right \rfloor + 1 +\min\{\left \lfloor \frac{x}{u} \right \rfloor, 1\} &\hbox{if}& stu \nmid b,\\
\end{array}
\right .$$
then $$\dm = n-k+1-\left (\left \lceil \frac{k}{r} \right \rceil - 1 \right) (\delta - 1);$$

\item If $b > a \geq \left \lceil \frac{k}{r} \right \rceil - 1$, $\left \lceil \frac{k}{r} \right \rceil = 2$ and $$\left \lceil \frac{n}{r+\delta- 1} \right \rceil \geq 
\left \{
\begin{array}{lcl}
\lceil \frac{b}{a}\rceil + 1 &\hbox{if}& 2a \leq r-1,\\
\left \lceil \frac{b}{\lfloor \frac{r-1}{2} \rfloor} \right \rceil + 1 &\hbox{if}& 2a > r-1,\\
\end{array}
\right . $$
then $$\dm = n-k+1-\left (\left \lceil \frac{k}{r} \right \rceil - 1 \right) (\delta - 1).$$
\end{enumerate}
\end{thm}

Just like in the remark below Theorem \ref{thm:max_d}, a simpler bound, but in general not as good, in Theorem \ref{thm:max_d_LRC}(iv) is $\left \lceil \frac{n}{r+\delta-1} \right \rceil \geq \left \lceil \frac{b}{stu} \right \rceil \left ( t (u-1) + 2 \right )$.

It was proven in \cite{kamath14} Corollary 2.3 that linear LRCs with all-symbol locality in the case when $r|k$ and $r+\delta - 1 \nmid n$ cannot achieve the Singleton-type bound given in Theorem \ref{thm:singleton_almost_affine}. This corresponds to the case $a = 0$ and $b>0$ in Theorem \ref{thm:non_existence_LRC}. Hence, by Theorem ~\ref{thm:max_d_LRC}~(ii), we obtain that 
$$
\dm=n-k+1 - \left (\left \lceil \frac{k}{r} \right \rceil - 1 \right )(\delta - 1) - 1,
$$
for linear $(n,k,d,r,\delta)$-LRCs when $r|k$ and $b=r+\delta-2$.

\section{Conclusions}
Recent progress in coding theory has proven matroid theory to be a valuable tool in many different contexts. This trend carries over to locally repairable codes. Especially the lattice of cyclic flats is a useful object to study, as its elements correspond to the local repair sets. 

We have thoroughly studied linear and more generally almost affine LRCs with all-symbol locality, as well as the connections of  these codes to matroid theory. We derived a generalized Singleton bound for matroids and nonexistence results for certain classes of $(n,k,d,r,\delta)$-matroids. These results can then be directly translated to nonexistence results for almost affine LRCs. 

Further, we have given several constructions of matroids with prescribed values of the parameters $(n,k,d,r,\delta)$. Using these matroid constructions, novel constructions of linear LRCs are given, using the representability of gammoids. Several classes of optimal linear LRCs then arise from these constructions. 

As future work, (non)existence results for matroids and linear and almost affine LRCs achieving the generalized Singleton bound remain open for certain classes of parameters $(n,k,r,\delta)$, when $\kr - 1 \leq a < b$. In addition, the size of the underlying finite field that our linear $(n,k,d,r,\delta)$-LRCs can be constructed over is left for future research. We expect that the upper bound $2^n$ arising from the related bound for all gammoids is loose for our class of matroids. We conjecture that all our matroids from Section~\ref{subsec:4_constructions_matroid} are representable over fields of size polynomial in $n$.

%\bibliographystyle{IEEE}	
%\bibliography{myrefs_special}	

\begin{thebibliography}{10}
\begin{comment}
\providecommand{\url}[1]{#1}
\csname url@rmstyle\endcsname
\providecommand{\newblock}{\relax}
\providecommand{\bibinfo}[2]{#2}
\providecommand\BIBentrySTDinterwordspacing{\spaceskip=0pt\relax}
\providecommand\BIBentryALTinterwordstretchfactor{4}
\providecommand\BIBentryALTinterwordspacing{\spaceskip=\fontdimen2\font plus
\BIBentryALTinterwordstretchfactor\fontdimen3\font minus
  \fontdimen4\font\relax}
\providecommand\BIBforeignlanguage[2]{{%
\expandafter\ifx\csname l@#1\endcsname\relax
\typeout{** WARNING: IEEEtran.bst: No hyphenation pattern has been}%
\typeout{** loaded for the language `#1'. Using the pattern for}%
\typeout{** the default language instead.}%
\else
\language=\csname l@#1\endcsname
\fi
#2}}
\end{comment}

\bibitem{tassie} T. Westerb\"ack, T. Ernvall, and C. Hollanti, ``Almost affine locally repairable codes and matroid theory,'' \emph{2014 IEEE Inf. Theory Workshop (ITW)}, 2014.

\bibitem{dimakis} A. G. Dimakis, P. B. Godfrey, Y. Wu, M. J. Wainwright, and K. Ramchandran,  ``Network coding for distributed storage systems,'' \emph{IEEE Trans. Inf. Theory}, 56(9), pp. 4539--4551, September 2010.

\bibitem{ghemawat03} S. Ghemawat, H. Gobioff, and S. T. Leung,
``The Google file system,`` In \emph{SOSP03, Proceedings of the nineteenth ACM symposium on Operating systems principles}, 2003.


\bibitem{me} T. Ernvall, S. El Rouayheb, C. Hollanti, and H. V. Poor, ``Capacity and security of heterogeneous distributed storage systems,'' \emph{IEEE J. Sel. Areas Comm.}, 31(12), pp. 2701--2709, Dec. 2013.

\bibitem{sasidharan} B. Sasidharan, and  P. V. Kumar,  ``High-rate regenerating codes through layering,'' \emph{2013 IEEE Int. Symp.  Inf. Theory (ISIT)}, pp. 1611--1615, 2013.

\bibitem{tian} C. Tian, V. Aggarwal, and V. A. Vaishampayan, ``Exact-repair regenerating codes via layered erasure correction and block designs,'' \emph{2013 IEEE Int. Symp.  Inf. Theory (ISIT)} ,  pp. 1431--1435, 2013.


\bibitem{toni} T. Ernvall, ``Codes between MBR and MSR points with exact repair property,'' \emph{IEEE Trans. Inf. Theory}, 60(11), pp. 6993--7005, Nov. 2014.

\bibitem{goparaju} S. Goparaju, S. El Rouayheb, and R. Calderbank, ``New codes and inner bounds for exact repair in distributed storage systems,''  \emph{2014 IEEE Int. Symp.  Inf. Theory (ISIT)}, pp. 1036--1040, 2014.

\bibitem{diskIO} I. Tamo, Z. Wang, and J. Bruck,  ``MDS array codes with optimal rebuilding,''  \emph{2011 IEEE Int. Symp.  Inf. Theory (ISIT)}, pp. 1240--1244, 2011.

\bibitem{Gopalan} P. Gopalan, C. Huang, H. Simitci, and S. Yekhanin,  ``On the locality of codeword symbols,'' \emph{IEEE Trans. Inf. Theory}, 58(11), pp. 6925--6934, 2012.

\bibitem{Oggier} F. Oggier and A. Datta,  ``Self-repairing homomorphic codes for distributed storage systems,'' \emph{2011 IEEE INFOCOM}, pp. 1215--1223.

\bibitem{Simple} D. S. Papailiopoulos, J. Luo, A. G. Dimakis, and C. Huang, J. Li  ``Simple regenerating codes: Network coding for cloud storage,'' \emph{2012 IEEE INFOCOM}, pp. 2801--2805.

\bibitem{LRCpapailiopoulos} D. S. Papailiopoulos, and A. G. Dimakis,  ``Locally repairable codes,'' \emph{2012 IEEE Int. Symp. Inf. Theory  (ISIT)}, pp. 2771--2775.

\bibitem{LRCmatroid} I. Tamo, D. S. Papailiopoulos, and A. G. Dimakis,  ``Optimal locally repairable codes and connections to matroid theory,'' \emph{2013 IEEE Int. Symp. Inf. Theory (ISIT)}, pp. 1814--1818.

\bibitem{simonis98} J. Simonis and A. Ashikhmin, ``Almost affine codes'',
\emph{Design, codes and cryptography}, 14, pp. 179--197, 1998.

\bibitem{whitney35} H. Whitney, `On the abstract properties of linear dependence'', \emph{Amer. J. Math}, 57, pp. 509--533, 1935.

\bibitem{greene76} C. Greene, `Weight enumeration and the geometry of linear codes'', \emph{Stud. Appl. Math}, 55, pp. 119--128, 1976.

\bibitem{macwilliams63} F. J. MacWilliams 'A theorem on the distribution of weights in a systematic code', \emph{Bell Syst. Tech J.}, 42, pp. 79--94, 1963.

\bibitem{barg97} A. Barg,
``The matroid supports of a linear code'', \emph{Appl. Algebra Engrg. Comm. Comput.}, 8, pp. 165--172, 1997.

\bibitem{britz10} T. Britz,
``Code enumerators and Tutte polynomials'', \emph{IEEE Trans. Inf. Theory}, 56, pp. 4350--4358, 2010.

\bibitem{kashyap08} N. Kashyap,
``A decomposition theory for binary linear codes'', \emph{IEEE Trans. Inf. Theory}, 54, pp. 3035--3058, 2008.

\bibitem{dougherty07} R. Dougherty, C. Freiling, and K. Zeger,
``Networks, matroids, and non-Shannon information inequalities'', \emph{IEEE Trans. Inf. Theory}, 53(6), pp. 1949--1969, 2007.

\bibitem{marti07} J. Mart\'i-Farr\'e and C. Padr\'o,
``On secret sharing schemes, matroids and polymatroids'', In S. Vadhan edt., \emph{4th Theory of Crypt. Conf. TCC 2007, Lecture Notes in Computer Science},  vol. 4392, pp. 253--272, 2007.

\bibitem{rouayheb10} S. El Rouayheb, A.Sprintson, and C. Georghiades, `On the index coding problem and its relation to network coding and matroid theory'', \emph{IEEE Trans. Inf. Theory}, 56(7), pp. 3187--3195, 2010.

\bibitem{Singleton} R. C. Singleton, ``Maximum distance $q$-nary codes'',
\emph{IEEE Trans. Inf. Theory}, 10, pp. 116--118, 1964.

\bibitem{prakash12} N. Prakash, G. M. Kamath, V. Lalitha, and P. V. Kumar,  
``Optimal linear codes with a local-error-correction property'',  \emph{2012 IEEE Int. Symp.  Inf. Theory (ISIT)}, pp. 2776 -- 2780.

\bibitem{prakash14} N. Prakash, V. Lalitha, and P. V. Kumar,  
``Codes with locality for two erasures'',  \emph{2014 IEEE Int. Symp.  Inf. Theory (ISIT)}, pp. 1962 -- 1966.

\bibitem{wang15} A. Wang, and Z. Zhang,  
``An integer programming-based bound for locally repairable codes''  \emph{IEEE Trans. Inf. Theory}, pp. 5280 -- 5294, 2015.

\bibitem{cadambe13} V. Cadambe and A. Mazumdar,
``An upper bound on the size of locally recoverable codes'', In \emph{Proc. IEEE Symp. Netw. Coding}, pp. 1--5, Jun. 2013.

\bibitem{rawat14} A. S. Rawat, O. O. Koyluoglu, N. Silberstein, and S. Vishwanath,
``Optimal locally repairable and secure codes for distributed storage systems'', \emph{IEEE Trans. Inf. Theory}, 60(1), pp. 212--236, 2014.

\bibitem{silberstein13} N. Silberstein, A. S. Rawat, O. O. Koyluoglu, and S. Vishwanath, 
``Optimal locally repairable codes via rank-metric codes,'' \emph{2013 IEEE Int. Symp. Inf. Theory  (ISIT)}, pp. 1819--1823.

\bibitem{song14} W. Song, S. H. Dau, C. Yuen, and T. J. Li,
``Optimal locally repairable linear codes'',\emph{IEEE J. Sel. Areas Comm.}, 32(5), pp. 1019--1036, 2014.

\bibitem{tamo14} I. Tamo and A. Barg, ``A family of optimal locally recoverable codes'' \emph{IEEE Trans. Inf. Theory}, 60(8), pp. 4661--4676, 2014


%\bibitem{britz05} T. Britz and C. G. Rutherford,
%``Covering radii are not matroid invariants'', \emph{Discrete Math.}, 296, pp. 117--120, 2005.

\bibitem{britz12} T. Britz, T. Johnsen, D. Mayhew, and K. Shiromoto,
``Wei-type duality theorems for matroids'', \emph{Designs, Codes and Cryptography}, 62, pp. 331--341, 2012.

\bibitem{eka-lrc} T. Ernvall, T. Westerb\"ack, C. Hollanti, and R. Freij, ``Constructions and properties of linear locally
repairable codes,'' submitted, arxiv:1410.6339.

\bibitem{lindstrom73} B. Lindstr{\"o}m,  ``On the vector representations of induced matroids'' \emph{Bull. London Math. Soc.}, 5, pp. 85--90, 1973.

\bibitem{stanley11} R. Stanley ``Enumerative combinatorics, vol 1''
\emph{2:ed Cambridge University Press}, 2011.

\bibitem{oxley92} J. Oxley, ``Matroid Theory''
\emph{Oxford Graduate Texts in Mathematics}, 3rd ed., Oxford University Press, 1992.

\bibitem{matus99} F. Mat\'u\^s,
``Matroid representation by partitions'', \emph{Discrete Math.}, 203, pp. 169--194, 1999.

\bibitem{bonin08} J. E. Bonin and A. de Mier,
``The lattice of cyclic flats of a matroid'', \emph{Annals of Combinatorics}, 12, pp. 155--170, 2008. 

\bibitem{shoda12} K. Shoda, ``Large families of matroids with the same Tutte polynomial'',
\emph{PhD thesis, The Geroge Washington University}, August 2012.

\bibitem{kamath14} Govinda M. Kamath, N. Prakash, V. Lalitha, and P. Vijay Kumar, ``Codes With Local Regeneration
and Erasure Correction'', \emph{IEEE Trans. Inf. Theory}, 60(8), pp. 4637--4660, 2014.

\bibitem{schrijver03} A. Schrijver, ``Combinatorial Optimization: Polyhedra and Efficiency. Vol B: Matroids, Trees, Stable Sets''
\emph{Algorithms and Combinatorics 24}, Springer-Verlag, 2003.

\bibitem{Hall} P. Hall, ``On Representatives of Subsets'',
\emph{J. London Math. Soc.},10, pp. 26--30, 1935. 

%\bibitem{brylawski86} T. H. Brylawski, 
%``Constructions'', \emph{Theory of Matroids} (edited by N. White), %Cambridge University Press, pp. 127--223, 1986.







%\bibitem{chan13} T. Chan, A. Grant and T. Britz,
%``Quasi-uniform codes and their applications'', \emph{IEEE Trans. Inf. Theory}, 59(12), pp. 7915--7926, December 2013.









%\bibitem{thomas13} E. K. Thomas and F. Oggier, 
%``Explicit constructions of quasi-uniform codes from groups'' \emph{2013 IEEE Int. Symp. Inf. Theory  (ISIT)}, pp. 489--493.




%\bibitem{TamoBarg} I. Tamo, A. Barg,  ``A family of optimal locally recoverable codes,'' \emph{arXiv:1311.3284}, 2013.

%\bibitem{Pyramid} C. Huang, M. Chen, J. Lin,  ``Pyramid codes: Flexible schemes to trade space for access efficiency in reliable data storage systems,'' \emph{IEEE Int. Symp. on Network Comp. and Appl.}, 2007, pp. 79-86.

%\bibitem{matroidbook} SOME BOOK


\end{thebibliography}

\newpage
\appendix

\section{Appendix}

\begin{proof}[\bf{Proof of Theorem \ref{theorem:(n,k,d,r,delta)-Z}}]
 For statement (i), we first claim that 
$$
\begin{array}{rcl}
d & = & \min \{|X| : X \subseteq E \hbox{ and } \rho(E \setminus X) < k\}\\
  & = & n - \max \{|Y| : Y \subseteq E \hbox{ and } \rho(Y) < k\}\\
  & = & n - \max \{|Y| : Y \in \mathcal{F} \setminus \{E\}\}\\
  & = & n - k + 1 - \max \{\eta(Y) : Y \in \mathcal{F} \setminus \{E\}\} \\
  & = & n - k + 1 - \max \{\eta(Z) : Z \in \coatom\} \\
\end{array}
$$
The first equality is the definition of $d$ from Definition \ref{def:(n,k,d,r,delta)-matroid}. The third equality claims that the maximum is obtained when $Y$ is a flat, which follows from Proposition \ref{prop:basic_facts} (ix) and the fact that $Y \subseteq \cl(Y)$. By maximality, $Y$ must have rank $\rho(Y) = k-1$, which gives the forth equality. The fifth equality now follows directly from Proposition \ref{prop:basic_facts} (iii).

For statement (ii), we first observe that $S_x$ in Definition \ref{def:(n,k,d,r,delta)-matroid} can be chosen to be cyclic, as we could otherwise find a smaller set with the same nullity and smaller size, by Proposition \ref{prop:basic_facts} (viii). Statements a) and b) in this Theorem follows directly from Definition \ref{def:(n,k,d,r,delta)-matroid}. Finally, statement c) in this theorem follows by applying (i) to $M \vert S_x$, observing that $1_{M \vert S_x} = S_x$. 
\end{proof}

\begin{proof}[\bf{Proof of Lemma \ref{lemma:fundamental_chain}}]

First, let  $\{S_x\}_{x \in E}$ be a collection of cyclic sets of $M$ for which the statements a) - c) in Theorem \ref{theorem:(n,k,d,r,delta)-Z}~(ii) are satisfied. 
We construct the chain $\{Y_j\}_{j=0}^m$ inductively by first letting $Y_0=\emptyset$. Given $Y_{j-1}\subsetneq E$, we choose $x_j\in E\setminus Y_{j-1}$ arbitrarily, and assign $Y_j=\mathrm{cl}(Y_{j-1}\cup S_x)$. If $Y_j=E$, we set $m=j$. 
 
 Let $j$ be any integer  in $\lbrack m \rbrack$ . We first observe that $\mathrm{cl}(S_j)$ is a cyclic flat, by Proposition \ref{prop:basic_facts} (iv). Hence, by Proposition \ref{pro:basic-Z} (iii), we see that inductively $Y_j$ is a cyclic flat with
\begin{equation} \label{eq:Y_j_join}
Y_{j} = \mathrm{cl}(Y_{j-1} \cup S_j) = \mathrm{cl}(Y_{j-1} \cup \mathrm{cl}(S_j)) = Y_{j-1} \vee \mathrm{cl}(S_j).
\end{equation} As $x_j\in Y_j\setminus Y_{j-1}$, we indeed have an increasing chain $$
C: 0_\mathcal{Z} = Y_0 \subsetneq Y_1 \subsetneq \ldots \subsetneq Y_m = E.
$$

We remark as in \eqref{eq:rank_S_x} that $\rho(S_j) \leq r$ for any $j \in \lbrack m \rbrack$. Hence, by axiom (R3) in \eqref{eq:rank_matroid}, we have
$$
\rho(Y_j) = \rho(Y_{j-1} \cup S_j) \leq \rho(Y_{j-1}) + \rho(S_j) \leq \rho(Y_{j-1}) + r.
$$

%For (ii), 
%the fact that $Y_{j-1} \subsetneq Y_j$ implies that $\rho(Y_{j-1}) < \rho(Y_j)$, as $Y_{j-1}$ and $Y_j$ are flats. Then, by Axiom (R3) in \eqref{eq:rank_matroid}, 
%$$
%\rho(S_j \cap Y_{j-1}) \leq \rho(S_j) + \rho(Y_{j-1}) - \rho(S_j \cup Y_{j-1}) = \rho(S_j) + \rho(Y_{j-1}) - \rho(Y_j) < \rho(S_j).
%$$
Moreover, by the statements (i) - (iii) in  Proposition~\ref{prop:basic_facts} and Theorem~\ref{theorem:(n,k,d,r,delta)-Z}~(ii)~c), we have
$$
\begin{array}{rcl}
\eta(Y_j) & = & \eta(\mathrm{cl}(Y_{j-1} \cup S_j))\\
          & \geq & \eta(Y_{j-1} \cup S_j) \\
					& \geq & \eta(Y_{j-1}) + \eta(S_j) - \eta(Y_{j-1} \cap S_j)\\
					& \geq & \eta(Y_{j-1}) + \eta(S_j) - \max \{ \eta(X) : X \in coA_{\mathcal{Z}(M \vert S_j)}\}\\
			  	& = & \eta(Y_{j-1}) + d(M \vert S_j) -1\\
					& \geq & \eta(Y_{j-1}) + \delta -1.
\end{array}
$$ This concludes the proof.
\end{proof}

\begin{proof}[\bf{Proof of Proposition \ref{proposition:bound_delta_k_rate}}]
For (i), let $Y$ be any subset of $E$ with $|Y| < \delta$. From Definition \ref{def:(n,k,d,r,delta)-matroid}, we conclude for every $x\in E$ that there is a subset $S_x \subseteq E$ with $x \in S_x$ and $\rho(S_x \setminus Y) = \rho(S_x)$. 

Hence $\rho(E \setminus Y) = \rho(E)$, since every flat containing $S_x\setminus Y$ must contain $S_x$, and $E = \bigcup_{x \in E} S_x$. Consequently, from the definition 
$
d = \min \{|X| : X \subseteq E \hbox{, } \rho(E \setminus X) < \rho(E)\}
$, it follows that $\delta \leq d$. 

For (ii), by (i) and Theorem \ref{th:bound}, 
$$
\delta \leq n - k + 1 - \left ( \left \lceil \frac{k}{r} \right \rceil  - 1\right ) (\delta - 1).
$$
Therefore
$$
k \leq n - \left \lceil \frac{k}{r} \right \rceil (\delta - 1).
$$

For (iii), by (ii), we have 
$$
\frac{n}{k} \geq 1 + \left \lceil \frac{k}{r} \right \rceil \frac {(\delta - 1)}{k} \geq 1 + \frac{\delta - 1}{r} =  \frac{r + \delta - 1}{r}.
$$

\end{proof}

%\begin{proof}[\bf{Proof of Lemma %\ref{lemma:fundamental_chain_bound_d}}]
%As $(\left \lceil \frac{k}{r} \right \rceil - 1) r < k = \rho(Y_m)$, %it follows from Lemma \ref{lemma:fundamental_chain} $(ii)$ that
%$$
%m - 1 \geq \left \lceil \frac{k}{r} \right \rceil - 1.
%$$
%Moreover, since $\rho(Y_{m-1}) < \rho(Y_m) = \rho(E)$, we have that
%$$
%\eta(Y_{m-1}) \leq \max \{\eta(Z) : Z \in coA_\mathcal{Z} \},
%$$
%by Proposition \ref{prop:basic_facts} (iii). Hence, by Theorem \ref{theorem:(n,k,d,r,delta)-Z}~(iii),
%$$
%d  =  n - k + 1 - \max \{ \eta(Z) : \eta(Z) \in coA_\mathcal{Z} \}  %\leq  n - k + 1 - \eta(Y_{m-1}).
%$$
%\end{proof}

\begin{proof}[\bf{Proof of Theorem \ref{th:structure-optimal-matroid}}]
Let 
\begin{equation} \label{eq:chain_to_singleton_bound}
C: 0_\mathcal{Z} = Y_0 \subsetneq Y_1 \subsetneq \ldots \subsetneq Y_m = E,
\end{equation}
be a chain of $(\mathcal{Z}(M),\subseteq)$ as given in Lemma \ref{lemma:fundamental_chain}~(i), from a subset $\{S_j\}_{j \in \lbrack m \rbrack}$ of $\{S_x\}_{x \in E}$. Since $d$ achieves the generalized Singleton bound in Theorem \ref{th:bound}, we get that $m = \left \lceil \frac{k}{r} \right \rceil$ and $\eta(Y_{m-1}) \leq (\left \lceil \frac{k}{r} \right \rceil - 1)(\delta - 1)$. Hence, 
\begin{equation} \label{eq:eta_Y_j}
\eta(Y_j) = j(\delta - 1) \hbox{ for } j = 0,1,\ldots,m-1, \quad \hbox{where} \quad m = \left \lceil \frac{k}{r} \right \rceil \geq 2,
\end{equation}
by Lemma \ref{lemma:fundamental_chain} (iii) and the proof given for Theorem \ref{th:bound}.

For (i), observe that $\eta(0_\mathcal{Z}) = \lvert 0_\mathcal{Z} \rvert$. Hence, if $0_\mathcal{Z} \neq \emptyset$, then $\eta(Y_0) = \eta(0_\mathcal{Z}) > 0$. This is a contradiction by \eqref{eq:eta_Y_j}. 

To prove (ii), first observe that for any $S_x$ we can select the chain in \eqref{eq:chain_to_singleton_bound}, such that $Y_1 =\cl(S_x)$. 

By \eqref{eq:eta_Y_j}, and since $\eta(X) \leq \eta(\mathrm{cl}(X))$ for any $X \subseteq E$, we get that
$$
\delta-1 = \eta(Y_1) \geq \eta(S_x).
$$
 Moreover, as we know from Theorem \ref{theorem:(n,k,d,r,delta)-Z}~(iv)~c),
$$
d(M | S_x) \leq \delta \iff \min \{ \vert X \rvert : X \subseteq S_x \hbox{, } \rho(S_x \setminus X) < \rho(S_x)\} \geq \delta,
$$
which implies that $\eta(S_x) \geq \delta - 1$, proving (ii) a). 

To prove  (ii) b), assume that $S_x$ is not a cyclic flat. Then 
$$
\eta(Y_1) = \eta(\mathrm{cl}(S_x)) > \eta(S_x) = \delta-1,
$$
which contradicts \eqref{eq:eta_Y_j}. Thus,
$$
0_{\mathcal{Z}(M \vert S_x)} = \emptyset  \hbox{, } 1_{\mathcal{Z}(M \vert S_x)} = S_x \hbox{ and } \mathcal{Z}(M \vert S_x) = \{X \in \mathcal{Z}(M) : X \subseteq S_x\}
$$
by Proposition \ref{prop:basic_facts} (vii). Now suppose there were a cyclic flat $Z \in \mathcal{Z}(M)$ such that $\emptyset \subsetneq Z \subsetneq S_x$. Then $\rho(Z) < \rho(S_x)$ and $\eta(Z) > \eta(\emptyset) = 0$ by axiom (Z2) in Theorem \ref{th:Z-axiom}. Consequently, by Proposition \ref{prop:basic_facts} (iii) and  Theorem \ref{theorem:(n,k,d,r,delta)-Z}~(iv)~(c),
$$
d(M \vert S_x) \leq \eta(S_x) + 1 - \eta(Z) \leq \delta - 1,
$$
contradicting the $(r,\delta)$-locality.

For (iii), we will first prove that any collection $F_1,\ldots,F_m$ of cyclic sets from $\{S_x : x \in E\}$ with a non-trivial union, and $m = \lceil \frac{k}{r} \rceil$, constitutes a chain as given in \eqref{eq:chain_to_singleton_bound}, with $Y_j = Y_{j-1} \vee F_j$ for $j=1,\ldots,m$.  
%Note by statement  (ii) b)  that $F_j$ is a cyclic flat for $j = 1, \ldots, m$.  Also, observe that
%\begin{equation} \label{eq:obserbation_Y_1}
%\emptyset = 0_\mathcal{Z} = Y_0 \subsetneq F_1 = Y_0 \vee F_1 = Y_1.
%\end{equation}
%Assume that for some $j \leq m-1$ we had 
%\begin{equation} \label{eq:assumption_vee}
%Y_l = F_1 \cup \ldots \cup F_l  \hbox{ and } Y_{l-1} \subsetneq Y_l \hbox{ for } l = 1,\ldots,j.
%\end{equation}
%Since $F_1, \ldots, F_m$ is a non trivial union of subsets it would follow that $F_{j+1} \nsubseteq Y_j$. 
Indeed, the chain in the proof of \eqref{eq:chain_to_singleton_bound} is obtained by sequentially choosing an arbitrary $S_x$ with $S_x\not\subseteq Y_j$, which can be chosen from $\{F_i\}$ as this is a collection with non-trivial union.

If $|Y_j \cap F_{j+1}| \geq \rho(F_{j+1})$, then we obtain that $
\mathrm{cl}(Y_j \cap F_{j+1}) = F_{j+1} \subseteq \mathrm{cl}(Y_j) = Y_j
$
by Proposition \ref{prop:basic_facts} (x). This is a contradiction, and consequently 
\begin{equation} \label{eq:Y_j_intersection_F_{j+1}}
|Y_j \cap F_{j+1}| < \rho(F_{j+1}).
\end{equation}
Now, by  (ii) b)  and Proposition \ref{pro:basic-Z} (ii), any subset $X \subseteq S_x$ contains a circuit if and only if $|X| > \rho(S_x)$, \emph{i.e.}, $\rho(X) = |X|$ if and only if $|X| \leq \rho(S_x)$. Consequently, $\eta(Y_j \cap F_{j+1}) = 0$. This implies, using Proposition \ref{prop:basic_facts} (ii), \eqref{eq:eta_Y_j} and statement (ii), that
\begin{equation} \label{eq:eta_Y_{j+1}}
\eta(Y_j \cup F_{j+1}) \geq \eta(Y_j) + \eta(F_{j+1}) - \eta(Y_j \cap F_{j+1}) = (j+1)(\delta - 1). 
\end{equation}
Furthermore by \eqref{eq:eta_Y_j}, if $j+1 \leq m-1$, then
$$
\eta(Y_j \cup F_{j+1}) \leq \eta(\mathrm{cl}(Y_j \cup F_{j+1})) = \eta(Y_{j+1}) = (j+1)(\delta - 1).
$$
Hence, $Y_{j+1} = F_1 \cup \ldots \cup F_{j+1}$ if $j+1 \leq m-1$. If $Y_m \neq E$, then $\rho(Y_m) < \rho(E)$ and $\eta(Y_m) > (\lceil \frac{k}{r} \rceil - 1)(\delta - 1)$. Then it follows, by Proposition \ref{prop:basic_facts} (iii) and  Theorem \ref{theorem:(n,k,d,r,delta)-Z}~(iii),  that 
$$
d < n - k + 1 - \left ( \left \lceil \frac{k}{r} \right \rceil - 1 \right )(\delta - 1).
$$
This is a contradiction. Consequently, $F_1,\ldots,F_m$ constitutes a chain as given in \eqref{eq:chain_to_singleton_bound}, with 
\begin{equation} \label{eq:Y_i_equals_union}
Y_j = \bigvee_{i=1}^j F_i = 
\left \{
\begin{array}{lll}
\bigcup_{i=1}^j F_i & \hbox{if} & j < \lceil \frac{k}{r} \rceil,\\
E & \hbox{if} & j = m = \lceil \frac{k}{r} \rceil.
\end{array}
\right .
\end{equation}

For statement (iii) c), we first notice that the statement follows directly from \eqref{eq:eta_Y_j} when $j < \lceil \frac{k}{r} \rceil$. When $j \geq m = \lceil \frac{k}{r} \rceil$ we conclude, using \eqref{eq:Y_i_equals_union}, that $\bigvee_{i=1}^j F_i = E$. Also, by \eqref{eq:eta_Y_j},  
$$
n - k \geq \eta(F_1 \cup \ldots \cup F_m) \geq \left \lceil \frac{k}{r} \right \rceil (\delta - 1).
$$ 
Statement (iii) d) follows directly from \eqref{eq:Y_i_equals_union}, and statement (iii) e) is a immediately consequence of (iii) c)--d). Statement (iii) f) follows from 
%fix number
\eqref{eq:Y_j_intersection_F_{j+1}} and (ii) a). 
\end{proof}

\begin{proof}[\bf{Proof of Theorem \ref{thm:construction}}] We will show that $\mathcal{Z}$ and $\rho$ define a matroid, by proving that the axioms (Z0)-(Z3) in Theorem \ref{th:Z-axiom} are satisfied by $\mathcal{Z}$ and $\rho.$  We let $I,J$ be two subsets of $[m]$.

(Z0) Since the collection of sets $F_1,\ldots,F_m$ has a non trivial union by 
assumption (iv), it follows that $F_I \subsetneq F_J$ if and only if $I \subsetneq J$. Hence, we immediately get that $\mathcal{Z}$ is a lattice under inclusion with 
$$
F_I \wedge F_J = F_{I \cap J} 
\quad \hbox{and} \quad
F_I \vee F_J =
\left \{
\begin{array}{lcl}
F_{I \cup J} & \hbox{if} & F_{I \cup J} \in \mathcal{Z}_{<k},\\
E & \hbox{if} & F_{I \cup J} \notin \mathcal{Z}_{<k},
\end{array}
\right .
$$
for $F_I, F_J \in \mathcal{Z}_{<k}$. Also, the bottom element in the lattice $0_\mathcal{Z}$ equals $\emptyset$ and by 
assumption (ii) in the top element $1_\mathcal{Z}$ equals $E$.

(Z1) Since $0_{\mathcal{Z}} = F_{\emptyset}$, we obtain that
$$
\rho(0_{\mathcal{Z}}) = \rho(F_{\emptyset}) = 0.
$$ 

(Z2) Since $F_I \subsetneq F_J$ if and only if $I \subsetneq J$, it is enough to prove that the axiom (Z2) holds for $F_I \subsetneq F_J$ in the following two cases:
$$
\begin{array}{rl}
(i) & F_J \in \mathcal{Z}_{<k} \hand J = I \cup \{j\} \hbox{ with } j \in [m]\setminus I,\\
(ii) & F_I \in \mathcal{Z}_{<k} \hand J = [m] \com \hbox{ i.e. }F_J = E. 
\end{array}
$$
In the first case, by 
the construction of $\rho$,  
$$
\begin{array}{rcl}
\rho(F_{I \cup \{j\}}) - \rho(F_I) & = & |F_{I \cup \{j\}}| - \sum_{l \in I \cup \{j\}} \eta(F_l) - (|F_I| - \sum_{i \in I} \eta(F_i))\\
&=& |F_j| - |F_j \cap F_I| - \eta(F_j)\\
&=& \rho(F_j) - |F_j \cap F_I|\\
&>& 0.
\end{array}
$$ 
Moreover, we have
$$
\begin{array}{l}
(|F_{I \cup \{j\}}| - \rho(F_{I \cup \{j\}})) - (|F_I| - \rho(F_I))  = \\
 \sum_{l \in {I \cup \{j\}}} \eta(F_l) -  \sum_{i \in I} \eta(F_i) =\\
\eta(F_j) > 0.
\end{array}
$$

For case (ii), we immediately get that $\rho(E) - \rho(F_I)=k-\rho(F_I) > 0$. Now, we claim that for any $j \in [m] \setminus \{I\}$ with $F_{I \cup \{j\}} \notin \mathcal{Z}_{<k}$, we have
\begin{equation} \label{eq:skrutt}
(|F_{I \cup \{j\}}| - k) - (|F_I| - \rho(F_I)) > 0. 
\end{equation}
By construction of $\mathcal{Z}_{<k}$,
$$
(|F_{I \cup \{j\}}| - k) - (|F_I| - \rho(F_I)) \geq \sum_{l \in I \cup \{j\}} \eta(F_l) - \sum_{i \in I}\eta(F_i) = \eta(F_j) > 0.
$$
Hence, by case (i) and \eqref{eq:skrutt}, it follows
$$
(|E| - \rho(E)) - (|F_I| - \rho(I)) > 0. 
$$

(Z3) Suppose that $F_I, F_J \in \mathcal{Z}_{<k}$. Then
$$
\begin{array}{l}
\rho(F_I) + \rho(F_J) - (\rho(F_I \vee F_J) + \rho(F_I \wedge F_J) + \lvert (F_I \cap F_J) \setminus  (F_I \wedge F_J)\rvert ) =\\
\rho(F_I) + \rho(F_J) - \rho(F_I \vee F_J) - \rho(F_{I \cap J}) - |F_I \cap F_J| +  |F_{I \cap J}|=\\
|F_I| - \sum_{i \in I} \eta(F_i) + |F_J| - \sum_{j \in J} \eta(F_j) - \rho(F_I \vee F_J) + \sum_{j \in I \cap J} \eta(F_j)   -  |F_I \cap F_J| =\\
|F_I \cup F_J| - \sum_{j \in I \cup J} \eta(F_j) - \rho(F_I \vee F_j) =\\
|F_{I \cup J}| - \sum_{j \in I \cup J} \eta(F_j) - \rho(F_I \vee F_j).
\end{array}
$$
If $F_{I \cup J} \in \mathcal{Z}_{<k}$, then
$$
|F_{I \cup J}| - \sum_{j \in I \cup J} \eta(F_j) - \rho(F_I \vee F_j) = |F_{I \cup J}| - \sum_{j \in I \cup J} \eta(F_j) - ( |F_{I \cup J}| - \sum_{j \in I \cup J} \eta(F_j) ) = 0.
$$
If $F_{I \cup J} \notin \mathcal{Z}_{<k}$, then $$
\begin{array}{rcl}
|F_{I \cup J}| - \sum_{j \in I \cup J} \eta(F_j) - \rho(F_I \vee F_j) &=& |F_{I \cup J}| - \sum_{j \in I \cup J} \eta(F_j) - k \\
&\geq& |F_{I \cup J}| - \sum_{j \in I \cup J} \eta(F_j) - ( |F_{I \cup J}| - \sum_{j \in I \cup J} \eta(F_j) )\\
&=& 0.
\end{array}
$$
Moreover, for $E$ and $F_I$ we have that
$$
\begin{array}{l}
\rho(F_I) + \rho(E) - (\rho(F_I \vee E) + \rho(F_I \wedge E) + \lvert (F_I \cap E) \setminus  (F_I \wedge E)\rvert ) =\\
\rho(F_I) + \rho(E) - \rho(E) - \rho(F_I) + |F_I| - |F_I| = 0.
\end{array}
$$

We have now proven that the axioms (Z0)--(Z3) in Theorem \ref{th:Z-axiom} are satisfied by $\mathcal{Z}$ and $\rho$. Hence, $\mathcal{Z}$ and $\rho$ define a matroid $M = M(F_1,\ldots,F_m;k;\rho)$ over $E$.

The parameters $(n,k,d,r,\delta)$ will now be investigated using Theorem \ref{theorem:(n,k,d,r,delta)-Z}. Firstly, the parameters $(n,k,d,r,\delta)$ are defined for $M$ with $n = |1_\mathcal{Z}| = |E|$ and $k = \rho(1_\mathcal{Z}) = \rho(E)$, since $E \in \mathcal{Z}$ and $\rho(E) = k > 0$. By Axiom $(Z2)$ in Theorem \ref{th:Z-axiom}, $\eta(Y) > \eta(X)$ for all $X,Y \in \mathcal{Z}$ when $X \subsetneq Y$. Hence, by Theorem \ref{theorem:(n,k,d,r,delta)-Z}~(iii), 
$$
\begin{array}{rcl}
d & = & n - k + 1 - \max \{\eta(F_I) : F_I \in \mathcal{Z}_{<k}\}\\
  & = & n - k + 1 - \max \{|F_I| - (|F_I | - \sum_{i \in I} \eta(F_i)) : F_I \in \mathcal{Z}_{<k}\}\\
	& = & n - k + 1 - \max \{\sum_{i \in I} \eta(F_i) : F_I \in \mathcal{Z}_{<k}\}.
\end{array}
$$
Let $\delta - 1 = \min_{i \in [m]} \{\eta(F_i)\}$ and $S$ be a subset of $F_i$ such that $|S| = \rho(F_i) + \delta - 1$. By construction and Proposition \ref{prop:basic_facts}~(vii),
$$
\mathcal{Z}(M | F_i) = \{Z \in \mathcal{Z}(M) : Z \subseteq F_i\} = \{\emptyset,F_i\}.
$$
Hence, from Proposition \ref{pro:basic-Z}~(i) and (ii),
$$
\rho(X) = \{|X|, \rho(F_i)\} \hbox{ for } X \subseteq F_i
$$
and
$$
\mathcal{C}(M) \cap F_i = \{X \subseteq F_i : |X| = \rho(F_i) + 1\}.
$$
This implies that $S$ is a cyclic set and that 
$$
d(M | S) = |S| - \rho(S) + 1 = \rho(F_i) + \delta - 1 - \rho(F_i) + 1 =  \delta.
$$ 
by Definition \ref{def:(n,k,d,r,delta)-matroid}~(iv)~c). Therefore, with $r = \max_{i \in [m]} \{\rho(F_i)\}$ and as $F_{[m]} = E$, statements (iv) a)--c) in Theorem \ref{theorem:(n,k,d,r,delta)-Z} are satisfied. Consequently, $M$ has $(r,\delta)$-locality and $S$ is a locality set. 

It remains to show that the independent sets $\mathcal{I}(M)$ equals $\mathcal{I}$.
We first point out that $$|F_J| - \sum_{j \in J} \eta(F_j) > |F_I| - \sum_{i \in I} \eta(F_i),$$ for $I \subsetneq J \subseteq [m]$.
Noting that $X \subseteq E$ is independent if and only if $X$ does not contain any circuits, and applying Proposition \ref{pro:basic-Z}~(ii), we get
$$
\begin{array}{lcl}
\mathcal{I}(M) & = & \{X \subseteq E : |X| \leq \rho(Y) \hbox{ for all } Y \in \mathcal{Z}\}\\
            & = & \{X \subseteq E : |F_I \cap X| \leq \min \{|F_I| - \sum_{i \in I} \eta(F_i), k \} \hbox{ for all } I \subseteq [m] \}\\
            &=& \mathcal{I}.
\end{array}
$$ 
%The right implication in statement (v) follows directly from statement (iv). For the proof of the left implication, assume that for some $I \subseteq [m]$ and $X \subseteq F_I$ that 
%$$
%|F_{I'} \cap X| \leq \min \{|F_{I'}| - \sum_{i \in I'} \eta(F_i), k \} \hbox{ for all } I' \subseteq I \}.
%$$
%Observe that, by statement (iii), $\min \{|F_{I'}| - \sum_{i \in I'} \eta(F_i), k \} = \rho(F_{I'})$. Now, choose an element $j \in [m] \setminus I$ and a subset $I' \subseteq I$. If $\rho(F_{I' \cup \{j\}}) = k$, then 
%$$
%|X \cap F_{I' \cup \{j\}}| \leq |X| = |X \cap F_{I}| \leq \rho(F_I) \leq k =  \rho(F_{I' \cup \{j\}}).
%$$
% Now suppose that $\rho(F_{I' \cup \{j\}}) < k$. Then, using statement \eqref{eq:set_construction_1}~(iv),
%$$
%\begin{array}{rcl}
%|X \cap F_{I' \cup \{j\}}| & \leq & |X \cap F_{I'}| + |F_I \cap F_j| - |F_{I'} \cap F_j|\\
%                           & \leq & \rho(F_{I'}) + \rho(F_j) - |F_{I'} \cap F_j|\\
%													 & = & |F_{I'}| - \sum_{i \in I'} \eta(F_i) + |F_j| - \eta(F_j) - |F_{I'} \cap F_j|\\
%													 & = & |F_{I' \cup \{j\}}| - \sum_{i \in (I \cup \{j\})} \eta(F_i)\\
%													 & = & \rho(F_{I' \cup \{j\}}).
%\end{array}
%$$ 
%The left implication of (v) now follows from (iv).  
\end{proof}

\begin{proof}[\bf{Proof of Theorem \ref{theorem:graph_1}}]
To prove the theorem, we will first show that the assumptions (i)--(iv) in Section \ref{sec:general} are satisfied by $(F_1,\ldots,F_m;k;\rho)$, obtained from the graph $(G,\gamma)$ in Section \ref{sec:special}. We will then show that the values of the parameters $(n,k,d,r,\delta)$ of $M(F_1,\ldots,F_m;k;\rho)$ are the ones requested in Theorem \ref{theorem:graph_1}.

Statement \ref{sec:general}~(i) follows directly from \ref{eq:graph_1}~(ii) and (iii).  \ref{sec:general}~(ii) is obvious. For \ref{sec:general}~(iii), we first notice that by \ref{eq:graph_1}~(iv) and (vi), we can define $\gamma$ as the size of a nonempty intersection of two sets $F_i$ and $F_j$ Hence, as $F_h\cap F_i\cap F_j=\emptyset$ for all $h,i,j\in [m]$, we know that

\begin{equation} \label{eq:graph_1_n} 
|F_{[m]}| = \sum_{i \in [m]} |F_i| - \sum_{w \in W} \gamma(w).
\end{equation}
Moreover, for $i \in [m]$, we have
$$
\eta(F_i) = |F_i| - \rho(F_i) = \delta - 1 + \beta(i).
$$
Consequently, 
$$
\begin{array}{rcl}
|F_{[m]}| - \sum_{i \in [m]} \eta(F_i) & = &\sum_{i \in [m]} |F_i| - m(\delta - 1) - \sum_{i \in [m]} \beta(i) - \sum_{w \in W} \gamma(w)\\
                                       & = & mr -  \sum_{i \in [m]} \alpha(i) -  \sum_{w \in W} \gamma(w).
\end{array}
$$
Therefore, by \ref{eq:graph_1}~(v),  \ref{sec:general}~(iii) holds. For \ref{sec:general}~(iv), we first remark that $$F_{[m] \setminus i} \cap F_i = \sum_{w = \{i,j\} \in W} \gamma(w)$$ and $\rho(F_i) = r - \alpha(i)$ for $i \in [m]$. Hence \ref{sec:general}~(iv) holds, by \ref{eq:graph_1}~(vi).

We will now determine the parameters $(n,k,d,r,\delta)$, proving that they agree for the graph and the matroid. The given parameters $(r,\delta)$ for the graphs also give $(r,\delta)$-locality of the matroid as $\rho(F_i) \leq r$ and $\eta(F_i) \geq \delta - 1$ by \eqref{eq:graph_1}~(ii) and (iii), and \ref{sec:special}~(i) and (ii). We have already proven that the parameter $k$ of the graph is the rank of the matroid. Moreover, by \eqref{eq:graph_1_n},
$$
n = |F_{[m]}| = | \sum_{i \in [m]} |F_i| - \sum_{w \in W} \gamma(w) = (r + \delta - 1)m - \sum_{i \in [m]} \alpha(i) + \sum_{i \in [m]} \beta(i) - \sum_{w \in W} \gamma(w).
$$  
The statement about $d$ in Theorem \ref{theorem:graph_1}~(ii) holds as a consequence of Theorem \ref{thm:construction}~(iii) and the properties that 
$$
\sum_{i \in I} \eta(F_i) = |I|(\delta - 1) + \sum_{i \in I} \beta(i)$$ and $$|F_I| - \sum_{i \in I} \eta(F_i) = r |I| - \sum_{i \in I} \alpha(i) - \sum_{w \subseteq I, w \in W \in I} \gamma(w).
$$
This concludes the proof.
\end{proof}

\noindent {\bf Proof of Theorem \ref{thm:max_d}:} We will divide the proof of Theorem \ref{thm:max_d} into the the parts (i)--(v). First, we recall that $a$ and $b$ are the integers where
$$
k = \left \lceil \frac{k}{r} \right \rceil r - a \hbox{ and } n = \left \lceil \frac{n}{r+\delta - 1} \right \rceil (r+ \delta -1) - b.
$$ 

%\begin{color}{red}
%MAYBE DELETE\\
%Moreover, for any two integers $i$ and $j$ where $i \leq j$, let 
%$$
%[i,j] = \{i,i+1,\ldots,j\}.
%$$
%

\begin{proof}[\bf{Proof of Theorem \ref{thm:max_d} (i)}]
We will mimic the proof of Theorem~\ref{thm:almosttight}, using the assumption that $a\geq b$ to tighten the bounds. Hence, let $m=\left\lceil\frac{n}{r+\delta-1}\right\rceil$, and let $F_1,\ldots F_{m-1}$ be disjoint sets with rank $r$ and size $r+\delta-1$. Let $F_m$ be disjoint from all of $F_1,\ldots F_{m-1}$, with size $$|F_m|=n-(m-1)(r+\delta-1)=r+\delta-1-b$$ and rank $\rho(F_m)=|F_m|-\delta+1=r-b$. Finally, let $M$ be defined by $\mathcal{Z}(M)=\{F_I\}$, where $F_I=\cup_{i\in I}F_i$, and $$\rho(F_I)=\min\{\sum_{i\in I}\rho(F_i), k\}.$$ 
Now, the union of any $\kr-1$ sets among $F_1,\ldots, F_{m-1}$, together with $F_m$, has rank $$r\kr -b=k+a-b\leq k.$$ Thus we have $\max\{|I|: \rho(F_I)<k\}= \kr$, so $M$ has minimum distance $$d= n-k+1-(\delta-1)\max\{|I|: \rho(F_I)<k\}= n-k- (\kr-1) (\delta - 1).$$\end{proof}

\begin{proof}[\bf{Proof of Theorem \ref{thm:max_d} (ii)}] We will use Theorem~\ref{thm:construction} to construct a $(n,k,d,r,\delta)$-matroid with $$d=n-k+1-\kr(\delta-1) +(b-r),$$ where $$b=(r+\delta-1)\left\lceil\frac{n}{r+\delta-1}\right\rceil-n>a=r\kr-k.$$ For this purpose,
let $m=\left\lceil\frac{n}{r+\delta-1}\right\rceil-1=\kr +t$, where $t\geq 0$ by Proposition~\ref{proposition:bound_frac_for_n}. Let $F_1,\ldots F_{m-1}$ be disjoint sets with rank $r$ and size $r+\delta-1$. Let $F_m$ be disjoint from all of $F_1,\ldots F_{m-1}$, with size $$|F_m|=n-(m-1)(r+\delta-1)=2(r+\delta-1)-b$$ and rank $\rho(F_m)=r$. Finally, let $M$ be defined by $\mathcal{Z}(M)=\{F_I\}$, where $F_I=\cup_{i\in I}F_i$, and $$\rho(F_I)=\min\{\sum_{i\in I}\rho(F_i), k\}.$$ 

Clearly, this has the desired values of $r$ and $n$. To guarantee that $M$ has rank $\rho(M)=\rho(F_{[m]})=k$ is the rank function of a matroid, we verify that
$$
k = \left \lceil \frac{k}{r} \right \rceil r - a = (m-t)r\leq mr = \sum_{i\in [m]}\rho(F_i).
$$
Statements \ref{eq:graph_1}~(vi) follows as $r-\alpha(i) \geq 1$ for $i \in [m]$ and as $G$ has no edges. Now, by Theorem~\ref{theorem:graph_1}, there is an $(n,k,d,r,\delta)$-matroid with 
$$
n = (r+ \delta - 1)m + (r + \delta - 1 - b) = (r+ \delta - 1)\left \lceil \frac{n}{r+ \delta - 1} \right \rceil  - b.
$$ 
Moreover, by Theorem \ref{thm:construction},
$$
d = n - k -1 - \left (\left \lceil \frac{k}{r} \right \rceil (\delta - 1) - (b - r) \right ),
$$
as 
$$
\max_{I \in V_{<k}} \{ (\delta - 1) |I| \} = \left (\left \lceil\frac{k}{r} \right \rceil - 1 \right ) + (r+\delta-1-b)
$$
for $$V_{<k} = \{ I \subseteq [m] : r |I| < k \}.$$
This concludes the proof.
\end{proof}

\begin{proof}[\bf{Proof of Theorem \ref{thm:max_d} (iii), right implication}]
By the structure theorem~\ref{th:structure-optimal-matroid}, we see that the existence of an $(n,k,d,r,\delta)$-matroid with $d= n-k + 1 - \left ( \left\lceil \frac{k}{r} \right\rceil - 1 \right ) (\delta - 1)$ implies the existence of subsets $F_1,\ldots, F_m$ of $[n]$ such that:
\begin{enumerate}[(i)]
\item $F_j \nsubseteq \bigcup_{i \in [m] \setminus \{j\}} F_i$ for  $j = 1,\ldots,m$,
\item $|F_i|\leq r+ \delta - 1$ for  $i = 1,\ldots,m$,
\item $|\bigcup_{i\in[m]} F_i|=n$,
\item $|F \cap (\bigcup_{i \in I} F_i) | \leq |F| - \delta$ for every $F \in \{F_i\}_{i \in [m] \setminus I}$  with $I \subseteq [m]$  and $|I| < \lceil \frac{k}{r} \rceil$,
\item $|\bigcup_{i\in I} F_i| - |I|(\delta - 1) \geq k = \left\lceil\frac{k}{r}\right\rceil r -a$ for every $I\subseteq [m]$ with $|I| \geq \left\lceil\frac{k}{r}\right\rceil$.
\end{enumerate}

For simplicity, denote $\lceil\frac{k}{r}\rceil=h$. For any set system $F_1,\dots, F_m$ where $|F_i|\leq r+\delta-1$ for every$i$, construct a graph $\mathcal{G}$ on vertex set $[m]$, with an edge between $i$ and $j$ if and only if $F_i\cap F_j \neq\emptyset$. Note that, when $I\subseteq[m]$ is such that the induced subgraph $\mathcal{G}[I]$ on $I$ is connected, then 
$$
|F_I|\leq (r + \delta - 2)|I|+1.
$$
If $\mathcal{G}[I]$ is connected and equality holds in the above inequality, then $I$ is said to be \emph{full}. Note that for every full component $I$ in $\mathcal{G}$ and integer $1 \leq u \leq |I|$, there is a subset $I' \subseteq I$ such that $|I'| = t$ and $I'$ is full. Denoting by $c(\mathcal{G}[I])$ the number of full components of $\mathcal{G}[I]$, we get 
$$
|F_I|\leq (r+\delta-2)|I|+c(\mathcal{G}[I]).
$$
Let $J$ be the union of the $h - a - 1$ largest full components of $\mathcal{G}$ together with all non-full components of $\mathcal{G}$. If $|J| \geq h$, then we have a subset of nodes $J' \subseteq J$ with $|J'| = h$, such that $c(\mathcal{G}[J']) \leq h - a - 1$. However, assuming $$|F_I|\geq h(r+\delta - 1)-a = h(r + \delta - 2) + h - a$$ for every subset $I\subseteq [m]$ with $|I|=h$, then $c(\mathcal{G}[J']) \geq h-a$. Hence, $|J| \leq h-1$ and  $\mathcal{G}[[m] \setminus J]$ is a union of full components $I_1,\ldots, I_s$ of $\mathcal{G}$, and these full components contain at most $\left\lfloor\frac{h-1}{h-1-a}\right\rfloor$ nodes each.

When bounding $\left \lceil \frac{n}{r+\delta - 1} \right \rceil$, we first notice that
\begin{enumerate}[(i)]
\item $|I|(r + \delta - 1) - |F_{I}| = |I| - 1$ if $I$ is connected and full,
\item $|I|(r + \delta - 1) - |F_{I}| \geq |I|$ if $I$ is connected and not full,
\item $h(r + \delta - 1) - |F_I| \leq a$ if $|I| \leq h$,
\item $|I_i|(r + \delta - 1) - |F_{I_i}| \leq \left\lfloor\frac{h-1}{h-1-a}\right\rfloor - 1$ for $1 \leq i \leq s$. 
\end{enumerate} 
Hence, 
$$ 
\begin{array}{rcl} 
b & = & m (r+\delta-1) - |F_{[m]}|\\   
  & = & |J|(r+\delta-1) - |F_J| + \sum_{i = 1}^s |I_i| (r+\delta-1) - |F_{I_i}|.	 
\end{array} 
$$ 
Also, as $|J| < h$, we get
$$ 
|J|(r+\delta-1) - |F_J| + \sum_{i = 1}^s |I_i| (r+\delta-1) - |F_{I_i}| \leq a + s \left ( \left\lfloor\frac{h-1}{h-1-a}\right\rfloor - 1 \right). 
$$ 
Hence, as $b > a$, we have $\left\lfloor\frac{h-1}{h-1-a}\right\rfloor \geq 2$, or equivalently $a \geq \left \lceil \frac{h}{2} \right \rceil$. Now, assume that $a \geq \left \lceil \frac{h}{2} \right \rceil$. For the cardinality of $F_{[m]}$ we have that 
$$ 
|F_{[m]}| = (|J| + \sum_{i = 1}^s |I_i|) (r + \delta - 1) -b. 
$$ 
Using (iii), (iv) and the property that $|J| < h$, we now obtain that
\begin{equation} \label{eq:bound_n_optimal}
\left \lceil \frac{|F_{[m]}|}{r+\delta-1} \right \rceil \geq h - 1 + \left \lfloor \frac{b-a}{\left\lfloor\frac{h-1}{h-1-a}\right\rfloor - 1} \right \rfloor\left\lfloor \frac{h-1}{h-1-a}\right\rfloor + t,
\end{equation}
where 
$$
t = \left \{
\begin{array}{l}
0  \hbox{ if } (\left\lfloor\frac{h-1}{h-1-a}\right\rfloor - 1) | (b-a)\\
(b-a) - \left \lfloor \frac{b-a}{\left\lfloor\frac{h-1}{h-1-a}\right\rfloor - 1} \right \rfloor \left ( \left\lfloor \frac{h-1}{h-1-a}\right\rfloor - 1 \right ) + 1 \hbox{ otherwise}.
\end{array}
\right.
$$
 Rearranging equation \eqref{eq:bound_n_optimal}, we find the bound 
\begin{equation} \label{eq:rearranging}
\left \lceil \frac{n}{r+\delta-1} \right \rceil \geq 
\left \lceil \frac{k}{r} \right \rceil -1 + \left(b-a\right)\left(1+\frac{1}{t}\right),
\end{equation}
where 
$$
t=\left\lfloor a/\left( h -1-a\right)\right\rfloor= \left\lfloor a/\left(\left \lceil \frac{k}{r} \right \rceil -1-a\right)\right\rfloor.
$$
\end{proof}

\emph{Construction 3:} To prove Theorem \ref{thm:max_d}~(iii), we will construct graphs $(G,\gamma)$ that satisfy the assumptions in \ref{sec:special} with $(k,r,\delta)$, and then use Theorem \ref{theorem:graph_2}. For simplicity, denote 
$$
s = \left \lfloor \frac{\kr  - 1}{\kr-1-a} \right \rfloor \hbox{, }  u = \lkr -1-a +\left \lceil \frac{b-a}{s-1} \right \rceil \hbox{ and } x = \lkr - 1 - s \left ( \lkr-1-a \right) 
$$
Let
\begin{equation} \label{eq:dmax_(iii)}
\begin{array}{rl}
(i) & \hbox{$m \geq \left \lceil \frac{k}{r} \right \rceil -1 + \left(b-a\right)\left(1+\frac{1}{t}\right)$,  where $t= \left\lfloor a/\left(\left \lceil \frac{k}{r} \right \rceil -1-a\right)\right\rfloor$,}\\
(ii) & \hbox{$G$ be the graph  consisting of vertex-disjoint paths $P_1,\ldots,P_u$ with}\\
&
     |P_i| = \left \{
		  \begin{array}{lcl}
			s + 1 &\hbox{if}& 1 \leq i \leq x,\\
			s &\hbox{if}& x+1 \leq i \leq u - 1,\\
			s &\hbox{if}& i = u \hbox{ and } s-1 \mid b-a,\\
			b - a - \lfloor \frac{b-a}{s-1} \rfloor (s-1) + 1 & \hbox{if}& i = u \hbox{ and } s-1 \nmid b-a,\\
			\end{array}
			\right . \\
(iii) & \hbox{$\gamma(w) = 1$ for each $w \in W$.}
\end{array}
\end{equation}

\begin{proof}[\bf{Proof of Theorem \ref{thm:max_d} (iii), left implication}]:
We first note that statement \ref{sec:special}~(i) follows directly as $G$ has no cycles. Statement \ref{sec:special}~(ii) is a consequence of \eqref{eq:dmax_(iii)}~(iii). For statement \ref{sec:special}~(iii), we first remark that by \eqref{eq:bound_n_optimal} and \eqref{eq:rearranging}, we get
$$
\sum_{w \in W} \gamma(w) = |W| =  (\sum_{i \in u} |P_i| ) - u = (s+1)x + (u-1-x)s + |P_u| - u = b.
$$
Statement \ref{sec:special}~(iv) follows directly from \eqref{eq:dmax_(iii)}~(i). Statement \ref{sec:special}~(v) follows from the fact that
$$
\sum_{i = 1}^{\kr - 1 - a} |P_i| = \lkr - 1$$ and $$\sum_{i,j \in P, w = \{i,j\} \in W} \gamma(w) = \lkr - 1 - (\lkr - 1 - a) = a,  
$$
where $P = \bigcup_{1 \leq i \leq \kr - 1 - a} P_i$. Finally, statement \ref{sec:special}~(vi) follows from the property that $\gamma(w) = 1$ for all $w \in W$. The result now follows using \eqref{eq:bound_n_optimal} and \eqref{eq:rearranging}, which imply that 
$$
|\bigcup_{i = 1}^u P_i| = \left \lceil \frac{k}{r} \right \rceil -1 + \left(b-a\right)\left(1+\frac{1}{t}\right),
$$
where $t= \left\lfloor a/\left(\left \lceil \frac{k}{r} \right \rceil -1-a\right)\right\rfloor$. 
\end{proof}

\emph{Construction 4:} To prove Theorem \ref{thm:max_d}~(iv), we will construct graphs $G =G(\gamma;k,r,\delta,a,b)$ that satisfy the statements in Corollary \ref{corollary:graph_2}, and then use Theorem \ref{theorem:graph_2}. For simplicity, denote 
$$
s = \left \lfloor \frac{a}{\lceil \frac{k}{r} \rceil - 1} \right \rfloor \hbox{, } t = \left \lfloor \frac{r-1}{s} \right \rfloor \hbox{, } u = \left \lceil \frac{\kr + 1}{2}\right \rceil \hbox{ and } x = \left \lceil \frac{b - \lfloor \frac{b}{stu} \rfloor stu}{s} \right \rceil.
$$
Before we are ready to construct $G$, we need some subgraphs that will be the building blocks of $G$. For $1 \leq i \leq t$, let $P_i$ denote a path containing $u + 1$ vertices, with $p_i$ as start vertex and $q_i$ as end vertex. Now, let $B$ denote the graph obtained from $\sqcup_i P_i$ by identifying all $p_i$ to the same vertex $p\in B$, all all the end vertices $q_i$ the same vertex $q\in B$ 

%and the internal vertices of the paths be distinct in $B$, \emph{i.e.},
%$$
%\begin{array}{rl}
%\bullet & p = p_1 = \ldots = p_t,\\
%\bullet & q = q_1 = \ldots = q_t,\\
%\bullet & v^{(i)} \in P_i \setminus \{p,q\} \hbox{, } v^{(j)} \in P_j \setminus \{p,q\} \hbox{ and } i \neq j \Rightarrow v^{(i)} \neq v^{(j)}.
%\end{array}
%$$ 
We will now define a subgraph of $B'(h)$ of $B$, where $h$ denotes the number of edges that the subgraph should have. First we remark that the number of edges in $B$ equals $tu$. Now, order the edges in $B$ from 1 to $tu$ by starting from the start vertex $p$ and ending in the end vertex $q$ for each path, ordering the edges path by path from $P_1$ to $P_t$. This is
$$
\begin{array}{rl}
\bullet & \hbox{the path $P_i$ is the sequence of vertices $p= v_1^{(i)},v_2^{(i)}, \ldots v_u^{(i)},  v_{u+1}^{(i)} = q$, then}\\
        & \hbox{edge $\{v_{j}^{(i)},v_{j+1}^{(i)}\}$ is ordered as edge number $(i-1)u + j$.}   
\end{array}
$$
The subgraph $B'(h)$ is now defined as the subgraph of $B$ that consists of the edges numbered from 1 to $x$ and the vertices associated to these edges. By $B'(0)$ we mean the graph with no vertices. 

%Before we are ready to construct our graph $G$, we need the number of vertices in $B$ and $B'(h)$. First we notice that t
The number of vertices of $B$ equals the number of internal nodes in paths $P_1,\ldots,P_t$ plus 2, \emph{i.e.}, 
$$
t(u-1) + 2.
$$
Moreover, the number of vertices in $B'(h)$, when $h \neq 0$, equals 
%the number of internal vertices in the full length $P_i$-paths contained in $B'(h)$ plus the number of edges not contained in these paths plus $|\{p,q\} \cap B'(h)|$, \emph{i.e.}, 
$$
\left \lfloor \frac{h}{u} \right \rfloor (u-1) + (h - \left \lfloor \frac{h}{u} \right \rfloor u) + (1 + \min\{\left \lfloor \frac{h}{u} \right \rfloor, 1\}) = h - \left \lfloor \frac{h}{u} \right \rfloor + 1 +\min\{\left \lfloor \frac{h}{u} \right \rfloor, 1\}.
$$  
Now, for the construction of $G$, let
\begin{equation} \label{eq:dmax_(iv)} 
\begin{array}{rl}
(i) & \hbox{$m \geq \left \lfloor \frac{b}{stu} \right \rfloor \left ( t (u-1) + 2 \right) + y,$ where}\\
&
y = 
\left \{
\begin{array}{lcl}
0 &\hbox{if}& stu \mid b,\\
x - \left \lfloor \frac{x}{u} \right \rfloor + 1 +\min\{\left \lfloor \frac{x}{u} \right \rfloor, 1\} &\hbox{if}& stu \nmid b;\\
\end{array}
\right .\\
(ii) & \hbox{$G$ be the graph with vertices $[m]$ and edges $W$, where $G$ consists of
   $\lfloor \frac{b}{stu} \rfloor$ copies}\\ 
   & \hbox{of $B$, one copy of $B'(x)$ and possibly some additional isolated vertices;}\\
(iii) & 
\begin{array}{rl}
\cdot & \hbox{If $s \mid b$ then $\gamma(w) = s$ for all $w \in W$,}\\
\cdot & \hbox{If $s \nmid b$ then}\\
& \gamma(w) =  \left \{
			\begin{array}{lcl}
			s &\hbox{if}& w \hbox{ is not the vertex number x in $B'(x)$},\\
			b - \lfloor \frac{b}{s} \rfloor s &\hbox{if}& w \hbox{ is the vertex number x in $B'(x)$};
			\end{array}
			\right .
\end{array}
\end{array}
\end{equation}

\begin{proof}[\bf{Proof of Theorem \ref{thm:max_d} (iv)}]
As $\lceil \frac{k}{r} \rceil \geq 3$, and the smallest size of a cycle in the graph is $2u \geq \kr +1 $, it follows that $G$ has no $l$-cycles for $l \leq \max \{3,\kr\}$. Also, the property that $1 \leq \gamma(w) \leq \left \lfloor \frac{a}{\kr-1} \right \rfloor$ for all edges $w$ in $G$ follows from \eqref{eq:dmax_(iv)}~(iii) as $s = \left \lfloor \frac{a}{\kr-1} \right \rfloor$.  For statement 

\ref{sec:special}~(iii), we remark that
for a copy of $B$ in $G$, the total sum of $\gamma(w)$ for all edges $w$ in $B$ equals $stu$. Moreover, the total sum of $\gamma(w)$ for all edges $w$ in $B'(x)$ equals $sx$ if $s | b$, and  $s(x-1) + b - \lfloor \frac{b}{s} \rfloor s$ if $s \nmid b$. Hence
$$
\sum_{w\in W} \gamma(w) = \left \lfloor \frac{b}{stu} \right \rfloor \sum_{edges \hbox{ } w \in B}  \gamma(w)+ \sum_{edges \hbox{ } w \in B'(x)} \gamma(w) = b.
$$ 
Statement 

\ref{sec:special}(vi) follows as
$
t s \leq r-1$ and $2s \leq 2 \left \lfloor \frac{a}{2} \right \rfloor \leq a \leq r-1$.
Hence, by Corollary \ref{corollary:graph_2} and Theorem \ref{theorem:graph_2}, the theorem is now proven.
\end{proof}

\emph{Construction 5 when $2a \leq r-1$:} We will construct graphs $(G,\gamma)$ that satisfy the statements in \ref{sec:special} and then use Theorem~\ref{theorem:graph_2}. To construct $G$, let 
\begin{equation} \label{eq:dmax_(v_1)}
\begin{array}{rl}
(a) & m \geq \left \lceil \frac{b}{a} \right \rceil + 1;\\
(b) & G \hbox{ be the graph with vertices $[m]$ and edges $W = \{ \{i,i+1\} : 1 \leq i \leq \left \lceil \frac{b}{a} \right \rceil\}$;}\\
(c) & \hbox{For $\{i,i+1\} \in W$ let,}\\
    & \gamma(\{i,i+1\}) = \left \{
		  \begin{array}{lcl}
			a &\hbox{if}& i < \lceil \frac{b}{a} \rceil,\\
			a  &\hbox{if}& i = \lceil \frac{b}{a} \rceil \hbox{ and } a | b,\\
			b - \lfloor \frac{b}{a} \rfloor a &\hbox{if}& i = \lceil \frac{b}{a} \rceil \hbox{ and } a \nmid b.
			\end{array}
		  \right .
\end{array}
\end{equation}

\begin{proof}[\bf{Proof of Theorem \ref{thm:max_d}(v) when $\mathbf{2a \leq r-1}$}]
That $G$ has no $l$-cycles for $l \leq \max \{3,\kr\}$ follows directly as $G$ has no cycles. Also, that $1 \leq \gamma(w) \leq \left \lfloor \frac{a}{\kr-1} \right \rfloor$ for all edges $w$ in $G$ follows directly from \eqref{eq:dmax_(v_1)}~(c). For statement \ref{sec:special}~(iii), we obtain from \eqref{eq:dmax_(v_1)}~(c) that
$$
\sum_{w \in W} \gamma(w) = \left (\frac{b}{a} - 1 \right )a + a = b \hbox{ if } a|b,
$$
and
$$
\sum_{w \in W} \gamma(w) = \left \lfloor \frac{b}{a} \right \rfloor a + b - \left \lfloor \frac{b}{a} \right \rfloor a = b \hbox{ if } a \nmid b.
$$
As the maximal number of neighbors of a vertex in $G$ is 2, we get that for any $i \in [m]$,
$$
\sum_{w = \{i,j\} \in W} \gamma(w) \leq 2a \leq r-1.
$$
Hence, by Corollary \ref{corollary:graph_2} and Theorem \ref{theorem:graph_2}, the theorem is now proven.
\end{proof}

\emph{Construction 5 when $2a > r-1$:} 
In order to prove Theorem \ref{thm:max_d}(v), we will construct graphs $(G,\gamma)$ that satisfy the statements in Corollary \ref{corollary:graph_2}, and then use Theorem~\ref{theorem:graph_2}. For simplicity, denote $h = \lfloor \frac{r-1}{2} \rfloor$. Now, to construct $G$, let 
%\begin{equation} %\label{eq:dmax_(v_2)}\end{equation}
\begin{enumerate}[(a)]
\item $m \geq \left \lceil \frac{b}{h} \right \rceil + 1$;
\item $G$ be the graph with vertices $[m]$ and edges $W = \{ \{i,i+1\} : 1 \leq i \leq \left \lceil \frac{b}{h} \right \rceil\}$;
\item For $\{i,i+1\} \in W$, let
$$ \gamma(\{i,i+1\}) = \left \{
		  \begin{array}{lcl}
			h &\hbox{if}& i < \lceil \frac{b}{h} \rceil,\\
			h  &\hbox{if}& i = \lceil \frac{b}{h} \rceil \hbox{ and } h| b,\\
			b - \lfloor \frac{b}{h} \rfloor h &\hbox{if}& i = \lceil \frac{b}{h} \rceil \hbox{ and } h \nmid b.
			\end{array}
		  \right .$$
\end{enumerate}

\begin{proof}[\bf{Proof of Theorem \ref{thm:max_d}(v) when $\mathbf{2a > r-1}$}]
The proof is completely analogous to the proof Theorem \ref{thm:max_d}(v) when $2a \leq r-1$, replacing $a$ by $h$.
\end{proof}

\begin{proof}[\bf{Proof of Lemma \ref{lemma:gammoid_graph}}]
We want to prove that the matroid $M(\bf{F})$ obtained from the set system $\bf F$ in Theorem~\ref{thm:construction} is isomorphic to the gammoid $M(G)$ associated to the graph $G$ in Algorithm 1. We will proceed by proving that the independent sets $\mathcal{I}(M(\bf F))$ and $\mathcal{I}(M(G))$ are equal.

The independent sets of $M(G)$ are
$$
\begin{array}{l}
\mathcal{I}(M(G)) = \{X \subseteq E : \exists \hbox{ a set of } \hbox{ $|X|$ vertex-disjoint paths from $X$ to $T$}\}.
\end{array}
$$
As each path from $E$ to $T$ goes through the complete bipartite graph between $H$ and $T$, with $|T|=k$, we can equivalently write 
$$
\mathcal{I}(M(G)) = \{X \subseteq E : \exists \hbox{ a matching of size } |X| \hbox{ between } X \hbox{ and } H, \hbox{ and }|X|\leq k\}.
$$

By Theorem \ref{thm:construction}, the independent sets of the matroid $M(\bf{F})$ are
$$
\begin{array}{lll}
\mathcal{I}(M({\bf F})) & = &  \{X \subseteq E: |X \cap F_I| \leq \min \{|F_I| - \sum_{i \in I} \eta(F_i), k \}\hbox{ for each } I \subseteq [m] \}\\
& = & \{X \subseteq E: |X \cap F_I| \leq |F_I| - \sum_{i \in I} \eta(F_i)\hbox{ for each } I \subseteq [m], \hbox{ and }|X|\leq k\}
\end{array}
$$

Hence, it remains to show that there is a matching of size $|X|$ in $G$ between $X$ and $H$, if and only if \begin{equation}\label{eq:sizes}|X \cap F_I| \leq |F_I| - \sum_{i \in I} \eta(F_i)\hbox{ for each } I \subseteq [m].\end{equation}
By Halls theorem~\cite{Hall}, a bipartite graph $(U,V,E)$ has a matching of size $U$ if and only if $|N(A)|\geq |A|$ for every $A\subseteq U$. Here, $U$ and $V$ are the two parts of the bipartition, and $$N(A)=\{x\in V : \exists u\in A \hbox{ with } ux\in E\}$$ is the neighborhood of $A$.

Assume that there is a matching of size $|X|$ in $G$ between $X$ and $H$. Then, in particular, $X \cap F_I$ has at least $|X \cap F_I|$ neighbors in $H$, for every $I\subseteq[m]$. But all these neighbors $v\in H$ must have $h(v)\cap I\neq \emptyset$, by construction of $G$. Now as $$|\{v\in H: i\in h(v)\}|=\rho(F_i)=|F_i|-\eta(F_i)$$ and $$|\{v\in H: \{i, j\}\subseteq h(v)\}|=|F_i\cap F_j|,$$ it follows by induction on $|I|$ that $$|\{u\in H : h(u)\cap I\neq \emptyset\}|=|F_I| - \sum_{i \in I} \eta(F_i).$$ This number of neighbors must be at least 
$|X \cap F_I|$, wherefore~\eqref{eq:sizes} holds.

Assume, on the other hand, that~\eqref{eq:sizes} holds, and let $A\subseteq X$ be an arbitrary subset of $X$. To apply Hall's Theorem, we need to prove that $|N(A)|\geq |A|$. 
%But $v\in N(A)$ if and only if there is $u\in A$ with $s(u)\subseteq h(v)$, so equivalently we can write $$N(A)=\cup_{u\in A}\{v\in H :s(u)\subseteq h(v)\}$$

Write $A = A' \cup A''$ where $A' = \{u \in A : |s(u)| = 1\}$ and  $A'' = \{u \in A : |s(u)| \geq 2\}$ respectively. For $x \in A''$, by (7--9) in Algorithm\ref{alg:transversal}, there is a node $u_x \in H$ for which $(\overrightarrow{x,u_x}) \in D$. Consequently, for $H'' = \{u_x \in H : x \in X''\}$, we have
\begin{equation} \label{eq:H''}
|H''| = |\{u_x \in H : x \in X''\}| = |X''|.
\end{equation}
Moreover, for $x \in A$, let $H_x = \{ u \in H : (\overrightarrow{x,u}) \in D\}$. By construction,
$$
|H_x| = |H_{s(x)}| = |F_{s(x)}| - \eta(F_{s(x)}).
$$
Hence, for $H' = \{u : \exists x \in A' \hbox{ such that } (\overrightarrow{x,u}) \in D\}$ and $I' = \{s(x) : x \in A'\}$, we get
\begin{equation} \label{eq:H'}
|H'| = |H_{I'}| = |F_{I'}| - \sum_{i \in I'} \eta(F_i).
\end{equation}
Since $A\subset X$, and $X$ satisfies ~\eqref{eq:sizes}, we know by \eqref{eq:H''} and \eqref{eq:H'}, we now obtain that
$
|A| \leq |H' \cup H''| \leq |N(A)|$. As $A\subset X$ was chosen arbitrarily, we can apply Hall's theorem to the effect that there is a matching between $X$ and $H$ of size $|X|$. This concludes the proof.
\end{proof}

\end{document}